\documentclass[12pt]{article}
\pdfoutput=1 
\usepackage[margin=1.25in]{geometry}
\usepackage{mathrsfs, mathtools}
\usepackage[title,toc,titletoc]{appendix}
\usepackage{titlesec}
\titleformat{\section}[block]{\large\normalfont\filcenter}{\thesection.}{.5em}{}
\titleformat{\subsection}[hang]{\bfseries}{\thesubsection.}{.5em}{}
\usepackage[dvipsnames]{xcolor}
\usepackage{bbm}
\usepackage{bm}
\usepackage{amsmath, amssymb, amsthm}	
\usepackage[shortlabels]{enumitem}
\usepackage[hang,flushmargin]{footmisc}
\usepackage[round]{natbib}
\usepackage[colorlinks=true,linkcolor = blue, citecolor=blue, hyperfootnotes=false, hyperindex,breaklinks]{hyperref}
\usepackage{cleveref}
\usepackage{bigints}
\usepackage{pgf}
\usepackage[justification=centering]{caption}
\usepackage{subcaption}
\usepackage{tikz}
\usetikzlibrary{patterns, decorations.pathreplacing}
\usepackage{hhline}
\usepackage{float}
\usepackage{siunitx}
\usepackage{textcomp}
\usepackage{graphics}
\usepackage{graphicx}

\newtheorem{proposition}{Proposition}

\newtheorem{definition}{Definition}
\newtheorem{lemma}{Lemma}
\newtheorem{assumption}{Assumption}

\DeclareMathOperator*{\argmax}{arg\,max}

\DeclareMathOperator*{\supp}{supp}
\DeclareMathOperator*{\mpc}{MPC}

\usepackage{setspace}
\usepackage{multirow}
\definecolor{shade}{gray}{.7}
\newcommand\blfootnote[1]{%
  \begingroup
  \renewcommand\thefootnote{}\footnote{#1}%
  \addtocounter{footnote}{-1}%
  \endgroup
}

\begin{document}
\title{How to Segment a Search Market:\\Information Design and Directed Search}
\author{Teddy Mekonnen\thanks{\textbf{Contact}: \href{mailto: mekonnen@brown.edu}{ mekonnen@brown.edu}. Department of Economics, Brown University. }}
\date{\today}
\maketitle
\begin{abstract}
This paper examines when the public provision of information in search markets improves welfare. I consider a two-sided frictional search market in which buyers match with vertically differentiated sellers. The market is segmented into submarkets based on seller types. Such segmentation serves as a public signal that buyers use to direct their search. Given a segmentation, I characterize both the socially efficient and the equilibrium allocation of buyers across submarkets, and identify a Hosios-type condition under which the equilibrium allocation is efficient. I then analyze the design of surplus-maximizing segmentations, showing that the nature of search externalities determines when the constrained-efficient segmentation fully separates seller types or pools them into at most a binary partition.
\end{abstract}

 \smallskip
 
\blfootnote{}
\blfootnote{I would like to express my deepest appreciation to Eddie Dekel, Federico Echenique, Bruno Strulovici, and Asher Wolinsky for their guidance and mentorship. This work has also greatly benefited from conversations with Alessandro Bonatti, Odilon C\^amara, Hector Chade, Yeon-Koo Che, Pedro Dal B\'o, Laura Doval, Jack Fanning, Kyungmin Kim, Stelios Michalopoulos, Bobak Pakzad-Hurson, Jo\~ao Ramos, Jesse Shapiro, Nikhil Vellodi, and Rajiv Vohra. Finally, I thank the editor and anonymous referees for comments and suggestions that greatly improved the paper. This paper supersedes a previous draft entitled “Random versus Directed Search for Scarce Resources.” }

\newpage

\allowdisplaybreaks

\section{Introduction}\label{sec:intro}
Agents in search markets frequently rely on publicly available information to guide their search for desirable goods and services. For example, firms often recruit job candidates from reputable universities;  shoppers on an online marketplace, such as \emph{Amazon}, filter for highly rated products; and  health insurance enrollees  consult review aggregators like \emph{Vitals} or \emph{Healthgrades} to identify top-rated physicians. In each of these settings, agents use information about the quality of goods or services to segment the market into smaller submarkets before  searching for a match in one of these submarkets. I refer to such information, which aids agents in directing their search prior to finding a match, as \emph{ex-ante information}.

In single-agent search environments, the social value of information is unambiguous: more precise ex-ante information improves efficiency by facilitating matches between agents and higher-quality goods and services. In real-world search markets, however, many agents compete for limited resources. Universities graduate a finite number of students each year, online marketplaces have limited inventories of products, and insurance networks feature a restricted pool of physicians. Moreover, these resources are typically vertically differentiated---firms prefer more productive employees, shoppers value higher-quality products, and patients seek more competent physicians---with agents vying for higher-quality matches. In such environments, agents' search decisions create externalities that influence the outcomes of others. Furthermore, these search externalities are shaped by the publicly available ex-ante information that guides the agents' strategic decisions. Thus, a central question arises: does the public provision of more precise ex-ante information enhance welfare in search markets with externalities?

This paper addresses this question by considering a two-sided market populated by a finite measure of buyers and sellers, where each buyer (he) has a unit demand, and  each seller (she) has a unit supply. The model has three defining features: First, sellers are vertically differentiated, with a seller's private type determining the quality of the good she supplies. Second, buyers observe a public signal about seller types, which they use to guide their search. Third, search is frictional, with the likelihood of meeting a trading partner determined by the market tightness---the ratio of buyers to sellers. The meeting function, which governs this frictional search, encapsulates two key externalities: a thick market externality, where a tighter market facilitates more meetings between buyers and sellers, and a congestion externality, where a tighter market intensifies competition among buyers, thereby reducing the probability that an individual buyer meets a seller. These opposing forces jointly determine equilibrium and welfare outcomes in search markets.

Public ex-ante information is modeled as a segmentation of the search market into a collection of submarkets. Sellers are allocated to these submarkets depending on their types, possibly with some randomization. I consider any flexible segmentation of the market \`a la \cite{ber15}. Based on this segmentation,  buyers form beliefs about the composition of seller types in each submarket. Each buyer then selects a submarket and engages in frictional search for a seller within it. If the buyer meets a seller, he observes her type, and the two bargain over the terms of trade.

The first set of results characterizes both the socially efficient and  equilibrium allocations of buyers across submarkets for any fixed market segmentation. From a welfare perspective, tighter submarkets are advantageous because they facilitate more meetings, thereby increasing the overall volume of trade. However, with a finite measure of buyers, raising tightness in one submarket necessarily reduces it elsewhere. The unique socially efficient allocation therefore balances the thick market externality created by tighter submarkets against the expected surplus generated within each submarket.

In contrast, the unique equilibrium allocation arises from how buyers trade off  the congestion in each submarket against the share of surplus they expect to capture in that submarket conditional on trade. In equilibrium, buyers follow a cutoff strategy based on a reservation value. They avoid submarkets where the expected payoff falls below this threshold. Among the remaining submarkets, those offering higher payoffs conditional on trade attract more buyers until their level of congestion precisely offsets their potential gains from trade.

Interestingly, a higher reservation value leads to lower buyer welfare. When buyers adopt a higher reservation value, they enter fewer submarkets, which increases the congestion within the submarkets they do join. This reduces the probability that any one buyer meets a seller, thereby lowering their average payoff.\footnote{This contrasts with single-agent search environments, where a higher reservation value does not generate congestion externalities and typically corresponds to a higher expected payoff.} In equilibrium, however, the reservation value itself must be consistent with buyers' average payoff, giving rise to a rational expectations condition.

This decentralized decision-making often leads to an inefficient allocation for two reasons: First, the seller types that generate the greatest total surplus may differ from those that provide buyers with the largest share of surplus, leading to a misalignment between the preferences of buyers and those of a social planner. Second, even when the preferences of the planner and the buyers coincide, buyers may over-concentrate in submarkets containing high-type sellers, creating excessive congestion. This misallocation highlights the dual role of information: while it helps buyers direct their search, it could also amplify the negative effects of congestion externalities.

Despite these two channels for inefficiencies, I identify a necessary and sufficient condition under which the socially efficient allocation of buyers is obtained as an equilibrium outcome. The condition, similar to \cite{hos90},  relates the efficiency of decentralized search to the elasticity of the meeting function and the bargaining power of buyers. When satisfied, this condition ensures that buyers internalize the congestion externalities they impose, aligning their private incentives with the planner’s objectives and yielding an efficient equilibrium outcome. However, this condition is satisfied only in knife-edge cases, implying that  equilibrium outcomes are inefficient in most practical settings.

The second set of results examines the design of surplus-maximizing market segmentations under two settings: one where buyers' allocations across submarkets is socially efficient (the first-best segmentation) and another where buyers' allocations reflect equilibrium behavior (the constrained-efficient segmentation). The first-best segmentation reflects the planner’s ability to directly control both the segmentation and the allocation of buyers, while the constrained-efficient segmentation accounts for buyers' strategic responses to the public provision of ex-ante information.

When buyers are efficiently allocated across submarkets, I show that a perfect market segmentation, where each seller type is assigned to a unique submarket, constitutes a first-best market segmentation. Intuitively, because more meetings lead to more trade, and because trade with a higher seller type generates a larger surplus, the socially efficient allocation of meetings must increase in seller types. Concentrating buyers into a single submarket for each seller type then allows the planner to achieve a desired volume of meetings while minimizing the number of buyers required. This segmentation enables the planner to tailor the allocation of meetings to seller types while fully leveraging the thick market externality.

Unlike the first-best setting, the constrained-efficient case requires the planner to account for buyers’ strategic responses to the ex-ante information revealed by the market segmentation. This interdependence between segmentation design and equilibrium behavior complicates the characterization of the optimal segmentation. Nevertheless, under mild assumptions on the bargaining power of buyers, I show that the problem can be decomposed into two distinct steps: an information design problem, where the planner determines what information to reveal about seller types, and a fixed-point problem, which ensures consistency between the segmentation and the buyers' equilibrium responses. This decomposition reveals that constrained-efficient market segmentations take on a monotone partition structure \citep{kle21, ari23}.

Further refinement to the structure of constrained-efficient segmentations arises under additional assumptions regarding the nature of congestion externalities. If congestion externalities worsen gradually as market tightness increases, perfect segmentation constitutes a constrained-efficient segmentation. In this case, a fully directed search market obtains, where buyers have complete ex-ante information about seller types. In contrast, if congestion externalities worsen rapidly, the constrained-efficient segmentation takes the form of a binary partition. Sellers are grouped into two broad categories: seller types above a threshold are pooled into a ``high" submarket, while types below the threshold are pooled into a ``low" submarket. In this case, search is partially-directed with all buyers joining only the high submarket and engaging in random search within it. 

These findings underscore the critical role of thick market and congestion externalities in determining the welfare consequences of more precise ex-ante information. When congestion externalities play a more limited role in search decisions, more informative public signals improve welfare, making perfect segmentation  constrained-efficient. Conversely, when congestion externalities dominate, the constrained-efficient segmentation withholds substantial information from the buyers, demonstrating the potential negative welfare implications from the public provision of more precise ex-ante information in search markets with externalities.

\subsection{Related Literature}

The role of information in search markets can be broadly divided into two categories. The first is ex-ante information, where agents observe signals prior to meeting a potential match and use this information to guide their search strategies. The second is interim information, where agents observe signals after meeting a potential match and use this information to determine the terms of trade.

This paper focuses on the former, and specifically, considers the role of ex-ante information in a market for search goods---goods whose quality is revealed upon inspection.\footnote{In this paper, a buyer observes a seller's type after the two meet, which reveals the quality of the seller's product. } In this context, \cite{anderson2006advertising}, \cite{choi2019optimal}, and \cite{lyu2023information} analyze a monopolist's optimal provision of information when buyers use it to decide between visiting the monopolist or taking their outside option. Similarly, \cite{choi2018consumer} explores the role of ex-ante information in oligopoly markets for horizontally differentiated search goods. However, these studies focus on settings with either a single buyer or sellers without capacity constraints, thereby abstracting away from search externalities. This paper addresses this gap by focusing on how ex-ante information interacts with  thick market and congestion externalities in markets with multiple buyers and capacity-constrained sellers.\footnote{Of course, other types of externalities could be considered. For example. \cite{vel18} considers the role of ex-ante information in the entry and exit decisions of sellers in a search market, and shows that suppressing information can incentivize the entry of new sellers and delay the exit of incumbents.}

This work also connects to the extensive literature on competitive search markets, where sellers post prices that buyers use to direct their search.\footnote{While the labor and matching literature often uses the  terms \emph{directed search} and \emph{competitive search} interchangeably, I follow \cite{wri21} in distinguishing between the two. As per their definition, ``Directed search means agents see some, although perhaps not all, characteristics of other agents, and based on that choose where to look for counterparties," while ``competitive search equilibrium means that agents on one side of the market post the terms of trade, while agents on the other side observe what is posted and direct their search accordingly." Thus, competitive search is a special case of directed search.} Examples include search for exchange goods \citep{but77, pet91}, labor markets \citep{mon91,moe97, morr02}, and markets with two-sided heterogeneity \citep{shi02, shim05, eec10}. A central insight from this literature is that equilibrium outcomes are typically efficient because posted prices internalize search externalities.

In contrast, this paper considers directed search driven solely by information rather than posted prices. Terms of trade are instead determined via bargaining once a buyer and seller meet,\footnote{One such environment is the labor market for high-skilled workers: firms often direct their search by recruiting candidates from top universities, but salaries are negotiated after a successful completion of the interview process. Similarly, the posted-price framework is unsuitable in environments where sellers lack commitment power, or where prices are fixed exogenously.}  as in \cite{dia82}, \cite{mor82}, and \cite{piss85}. Consequently, prices cannot fully internalize the search externalities  unless a Hosios condition \citep{hos90} is satisfied. Hence, this paper demonstrates that beyond price-based competitive search, directed search  can give rise to inefficient equilibrium outcomes. Furthermore, it identifies conditions under which information design can mitigate these inefficiencies, offering insights into the optimal provision of ex-ante information in search markets.

Few papers explore directed search beyond price-based competitive search. The most relevant work in this regard is \cite{men07}, who examines a private-values labor market where informed firms send cheap-talk messages that uninformed workers use to direct their search. Once a trading partner is met, the terms of trade are determined via bargaining. \cite{men07} shows that while informative cheap-talk equilibria exist, all equilibrium outcomes are inefficient.\footnote{When the trading mechanism is a first-price auction rather than bargaining, \cite{kim2015efficient} demonstrates the existence of a fully revealing and efficient cheap-talk equilibrium in a private-values setting, whereas \cite{kim2012endogenous} shows that no such equilibrium exists in a setting with vertical differentiation.} These inefficiencies stem from two interdependent externalities: search externalities on the workers' side of the market, and signaling externalities on the firms' side, which arise because the informational content of a firm's cheap-talk message depends on the messaging strategies of other firms.

In this paper, however, the market segmentation (and thus, its informational content) is taken as exogenously given by both buyers and sellers, reflecting the increasingly data-driven segmentation of markets based on certification or reviews rather than cheap talk. As a result, signaling externalities are absent from my model, leading to a clean analysis of how information interacts with search externalities. This, in turn, allows me to identify conditions under which efficient and equilibrium outcomes coincide. 

Of course, directed search need not be based solely on posted prices or information; rather, it may be based on both. \cite{eeckhout2010sorting} show that when buyers and sellers meet bilaterally, it is efficient to perfectly segment the market by type, and that this outcome can be decentralized by having buyers direct their search based on both the market segmentation and sellers' posted prices. Similarly, \cite{cai2017search} demonstrate that bilateral meetings are both sufficient and necessary for perfect market segmentations to emerge as socially efficient. In this respect, my finding that perfect segmentation is efficient is not surprising. However, I show that this outcome can be implemented as an equilibrium if and only if the Hosios condition identified in the paper is satisfied.

Finally, this paper loosely relates to the body of work on the role of interim information, especially in markets for \emph{experience goods}---goods whose quality is revealed only upon consumption. In the context of large search markets, \cite{lau12} and \cite{les16} demonstrate the negative welfare implications of interim information provision. A growing body of work also considers the design of interim information for experience goods in single-agent settings: \cite{rom19}, \cite{dogan2022consumer}, \cite{hu2022industry}, \cite{mekonnen2023persuaded}, and \cite{sato2023persuasion} all consider centralized provision of interim information, while \cite{board2018competitive}, \cite{au2023attraction}, \cite{mekonnen2024competition}, and \cite{he2023competitive} consider competitive provision of information. Beyond search markets, the welfare implications and the design of interim information has also been analyzed in various other contexts, including price-discriminating monopoly \citep{ber15}, differentiated Bertrand competition \citep{elliott2021market}, bilateral matching markets \cite{condorelli2023buyer}, and many-to-one matching markets \cite{li2023mismatch}.

The remainder of the paper is structured as follows:  I describe the model in \Cref{model}, and derive the first-best outcome in \Cref{sec:first-best}. I then characterize the equilibrium outcome for a fixed market segmentation in \Cref{sec:equilibrium}. Finally, \Cref{sec:equil-efficiency} discusses the efficiency properties of equilibrium outcomes and tackles the design of a constrained-efficient market segmentation.  All proofs are in the \hyperref[appendix]{Appendix}.

\section{Model}\label{model}

A two-sided market is populated by a unit mass of homogeneous buyers on one side and a mass $k>0$ of heterogeneous sellers on the other. Each buyer (he) has a unit demand for a good, and each seller (she) has a unit supply. A seller's type, denoted by $\theta$, captures the quality of her good. The type space $\Theta$ is a compact interval normalized to $[0,1]$. Let $\mathcal{B}(\Theta)$ denote the Borel $\sigma$-algebra of $\Theta$ and $\Delta(\Theta)$ denote the space of probability measures over $(\Theta, \mathcal{B}(\Theta))$. Types are distributed according to an absolutely continuous probability measure $\mu\in\Delta(\Theta)$. Each seller knows her own type, but each buyer ex-ante knows only the type distribution; he observes a seller's type only upon meeting her.

The market is segmented into submarkets, with the segmentation taken as exogenous and common knowledge among the buyers and sellers. Each seller is allocated to a submarket based on her type, possibly with some randomization.  Each buyer, on the other hand, can freely choose which submarkets to join. 
  
The timing is as follows: Given the market segmentation, each buyer picks a submarket to join. Within each submarket, buyers and sellers meet each other according to a \emph{bilateral meeting function}, which will be formalized shortly. Upon meeting a seller, a buyer observes the seller's type, and the two bargain over the terms of trade. If they reach a mutual agreement to trade at some price $p\in \mathbb{R}$, the buyer and seller obtain payoffs $\theta-p$ and $p$, respectively. Buyer-seller pairs that fail to trade as well as buyers and sellers who do not find a trading partner earn a payoff of zero.

\subsection{Market segmentation}

A submarket is simply a market populated by a subset of seller types.  Accordingly, each submarket can be indexed by the distribution of seller types it contains; that is, a submarket with seller-type distribution $\nu \in \Delta(\Theta)$ is referred to as submarket $\nu$.

A market segmentation is an allocation of sellers to submarkets based on their types. Formally, let $\mathcal{B}(\Delta(\Theta))$ denote the Borel $\sigma$-algebra over $\Delta(\Theta)$, and let $\Delta(\Delta(\Theta))$ be the space of probability measures over $(\Delta(\Theta), \mathcal{B}(\Delta(\Theta)))$. A market segmentation is defined as a stochastic kernel $\sigma: \mathcal{B}(\Delta(\Theta)) \times \Theta \to [0,1]$ such that,
\begin{enumerate}[$(a)$]
\item For each $\theta \in \Theta$, the mapping $B \mapsto \sigma(B, \theta)$ is a probability measure in $\Delta(\Delta(\Theta))$;
\item For each $B \in \mathcal{B}(\Delta(\Theta))$, the mapping $\theta \mapsto \sigma(B, \theta)$ is $\mathcal{B}(\Theta)$-measurable.
\end{enumerate}
Under segmentation $\sigma$, a type-$\theta$ seller is allocated to submarket $\nu$  with probability $\sigma(d\nu,\theta)$. Consequently, out of the total mass $k$ of sellers, the fraction allocated to submarket $\nu$ is given by
\[
\sigma_\mu(d\nu)\coloneqq\int_\Theta\sigma(d\nu, \theta)\mu(d\theta),
\]
and the support $\supp(\sigma_\mu)$ represents the submarkets present in the market.

Of course, the market segmentation must ensure that each submarket $\nu$ actually contains seller types distributed according to $\nu$. In particular, $\sigma$ must satisfy the following consistency condition: for all $A\in \mathcal{B}(\Theta)$ and $B\in \mathcal{B}(\Delta(\Theta))$,
\[
\label{eq:bayes}
\tag{1}
\int_A \sigma(B, \theta)\mu(d\theta)=\int_{B} \nu(A)\sigma_\mu(d\nu).
\]
This identity has a natural interpretation: consider sellers   types in some measurable subset $A\subseteq \Theta$. The fraction of these types allocated to submarkets in $B$, given by the left-hand side, must equal the fraction of the same types contained in submarkets in $B$, given by the right-hand side.

The consistency condition in \eqref{eq:bayes} guarantees that each submarket $\nu$ indeed has seller-type distribution $\nu$, and that buyers can infer this using Bayes’ rule. Thus, submarkets serve a dual role for buyers: each one is a platform for meeting a seller, and it also serves as a signal of the seller types therein. This is especially transparent when evaluating the consistency condition \eqref{eq:bayes} at $B=\Delta(\Theta)$, which yields
 \[
\int_{\Delta(\Theta)}\nu\sigma_\mu(d\nu)=\mu.
\]
This coincides with the definition of a market segmentation in \cite{ber15}, and the Bayes-plausibility constraint in \cite{kam11}.

\subsection{Bilateral Meeting functions}\label{mf}
Meetings are assumed to be bilateral, so a buyer meets at most one seller and vice versa. Following the literature on search and matching \citep{pet01, rog05},  the likelihood that a buyer and seller meet each other within a submarket depends only on the \emph{submarket tightness}, which is the ratio of buyers to sellers in that submarket. 

Formally, a submarket populated by a measure $b$ of buyers and a measure $s$ of sellers has a submarket tightness of $t\coloneqq b/s\in\mathbb{R}_+\cup\{\infty\}$. Each seller meets a buyer with probability $m(t)$ and meets no one with probability $1-m(t)$, where the mapping $t\mapsto m(t)$ denotes the \emph{meeting function}. As a measure $s\cdot m(t)$ of sellers meet buyers and meetings are bilateral, there must also be a measure $s\cdot m(t)$ of buyers that meet sellers. Thus, each buyer meets a seller with probability $m(t)/t$. The following assumptions on the meeting function are maintained for the remainder of the paper.
\medskip

\begin{assumption}
\label{ass:meeting-function} The meeting function satisfies the following:
\begin{enumerate}[$(\alph*)$]
\item $m(t)\leq \min\{1, t\}$ for all $t\geq 0$, and
\item  $m$ is twice differentiable, strictly increasing, and strictly concave.
\end{enumerate}
\end{assumption}
\medskip

The first assumption implies that both $m(t)$ and  $m(t)/t$ are well-defined probabilities. To motivate the second assumption, notice that as the submarket tightness increases, the market contains even more buyers relative to sellers, which has two effects: On the one hand, each seller in the submarket has more opportunities to meet a buyer, giving rise to a \emph{thick market externality} on the sellers' side. This effect is captured by the monotonicity of the meeting function.\footnote{Some earlier work in the matching and search literature uses the term ``thick market externality” to refer to increasing returns to scale in the meeting function, which can give rise to  multiple equilibrium outcomes in search markets \citep{diamond1982aggregate}. In contrast, and consistent with much of the recent literature, I assume constant returns to scale. See \citet{pet01} for a more in-depth survey.} On the other hand, each buyer in the submarket faces more competition to meet a seller, giving rise to a \emph{congestion externality} on the buyers' side. This effect is captured by the strict concavity of the meeting function, which guarantees that the mapping $t\mapsto m(t)/t$ is strictly decreasing. 

Let $\alpha\coloneqq m(\infty)$ be a seller's highest probability of meeting a buyer, and let  $\beta\coloneqq  m(0)/0$ be a buyer's highest probability of meeting a seller, with $\alpha, \beta\in (0,1]$. In other words, an agent's probability of meeting a trading partner is highest in an extremely unbalanced submarket where the agent is on the scarce side. However, even in such extreme cases, the agent is not guaranteed a meeting if $\alpha<1$ (for a seller) or $\beta<1$ (for a buyer).

Examples of meeting functions that satisfy \autoref{ass:meeting-function} include the \emph{CES meeting function}
\[
m(t)=\frac{\alpha \beta t}{\left(\alpha^\rho+(\beta t)^\rho\right)^{1/\rho}}
\]
with $\rho>0$, and the \emph{Urn-ball meeting function}
\[
m(t)=\beta t\left(1-\exp\left(\frac{-\alpha}{ \beta t}\right)\right).
\]
While \autoref{ass:meeting-function} is common in the search and matching literature, it is admittedly restrictive. However, by  standard continuity arguments, the results of this paper extend to meeting functions that do not satisfy \autoref{ass:meeting-function}, provided such functions can be approximated by meeting functions that do satisfy it. For example, consider the meeting function $m(t)=\min\{\alpha, \beta t\}$, which generalizes the \emph{frictionless meeting function} given by $m(t)=\min\{1, t\}$. This meeting function does not satisfy  Point $(b)$ of  \autoref{ass:meeting-function}, as it is only weakly increasing, weakly concave, and non-differentiable at $t=\alpha/\beta$. Nevertheless, the CES meeting function approximates the frictionless one as $\rho\to\infty$, and thus, the insights of this paper extend to this class of generalized frictionless meeting functions.

\subsection{Trade decisions}\label{td}
I assume that a buyer and seller negotiate the terms of trade only after they meet, and that their negotiations result in all ex-post efficient trades being realized. Specifically, suppose a buyer has met a type-$\theta$ seller. If they trade, they generate a total surplus of $\theta\in [0,1]$; If they do not, the surplus is zero. Thus, they mutually agree to trade at some price $p(\theta)$, which  requires $\theta\geq p(\theta)$ for the buyer to agree and $p(\theta)\geq 0$ for the seller to agree. In other words, the price $p(\theta)$ is some convex combination of the surplus $\theta$ and the zero outside option. 

Furthermore, I assume that buyers capture a non-negligible share of the surplus from trading with some seller types; otherwise, the buyers would have no incentives to direct their search to any specific submarket, making the problem at hand trivial. Formally, I assume:
\medskip

\begin{assumption}
\label{ass:2}
There exists a $\mathcal{B}(\Theta)$-measurable function $\lambda:\Theta\to [0,1]$ such that 
\begin{enumerate}[$(\alph*)$]
\item \textbf{Ex-post efficiency:} $p(\theta)=\big(1-\lambda(\theta)\big)\theta$ for all $\theta\in\Theta$, and 
\item \textbf{Non-triviality:} $\mu(\{\theta\in\Theta:\lambda(\theta)>0\})>0$.
\end{enumerate}
\end{assumption}
\medskip

\noindent I refer to $\lambda$ as the \emph{surplus-splitting function}, with  $\theta-p(\theta)=\lambda(\theta)\theta$ representing the share of surplus captured by a buyer when he trades with a type-$\theta$ seller.

One foundation for $p(\cdot)$ is as the outcome of a Rubinstein bargaining problem: upon meeting, a buyer and a seller enter a ``bargaining phase" in which the pair make alternating offers on how to divide the surplus. When buyers have a discount factor of $\varrho_b\in (0,1)$ and type-$\theta$ sellers have a discount factor of $\varrho_\theta\in (0,1)$, the unique solution to the bargaining phase yields a surplus-splitting function of $\lambda(\theta)=(1-\varrho_\theta)/(1-\varrho_b\varrho_\theta)$ when buyers make the first proposal, or $\lambda(\theta)=(1-\varrho_b)/(1-\varrho_b\varrho_\theta)$ when sellers make the first proposal. 

The price $p(\theta)$ can also be rationalized as the solution to a generalized  Nash-bargaining problem
\[
\max_{p\in\mathbb{R}}\big(\theta-p\big)^{\lambda(\theta)}p^{1-\lambda(\theta)},
\]
with $\lambda(\theta)$ capturing the buyer's bargaining power when he meets a type-$\theta$ seller. 

Finally, $p(\theta)$ can also reflect the expected price in a setting where a buyer who, upon meeting a type-$\theta$ seller, makes a take-it-or-leave-it offer with probability $\lambda(\theta)$, while the seller does the same with the complementary probability. If $\lambda(\theta)=1$ for all $\theta\in\Theta$, the model becomes analogous to one with non-transferable utility because a buyer can trade with any seller he meets at a price of zero.

\section{First-best market segmentation}\label{sec:first-best}
As a benchmark, consider a planner that seeks to maximize total surplus by choosing the market segmentation as well as the allocation of buyers across submarkets. A market segmentation $\sigma$ determines the allocation of sellers, with $k\sigma_\mu(B)$ representing the mass of sellers allocated to submarkets in $B\in\mathcal{B}(\Delta(\Theta))$. Let $Q\in\Delta(\Delta(\Theta))$ be the buyer-allocation policy, with $Q(B)$ representing the mass of buyers allocated to submarkets in $B\in\mathcal{B}(\Delta(\Theta))$.

Buyers are allocated only to submarkets that contain sellers, so  $Q$ must be absolutely continuous with respect to $\sigma_\mu$. In particular, there exists a measurable function $\tau:\Delta(\Theta)\to \mathbb{R}_+$  such that $Q(d\nu)=k\tau(\nu)\sigma_\mu(d\nu)$  for $\sigma_\mu$-almost all $\nu\in\Delta(\Theta)$. In other words, the mapping $\nu\mapsto \tau(\nu)$ represents the ratio of buyers to sellers within each submarket and, thus, I refer to it as the \emph{submarket tightness function}.  

The market segmentation and the submarket tightness function together fully determine the buyer-allocation policy. Hence, given a market segmentation $\sigma$, choosing a buyer-allocation policy $Q$ is equivalent to choosing a submarket tightness function $\tau$ that satisfies the following feasibility constraint:
\[
\label{eq:feasible}
\tag{2}
k\int_{\Delta(\Theta)}\tau (\nu)\sigma_\mu(d\nu)=1.
\]

Given market segmentation $\sigma$ and submarket tightness function $\tau$, there is a measure $k\sigma_\mu(d\nu)m(\tau(\nu))$ of meetings in submarket $\nu\in\supp(\sigma_\mu)$, with each meeting generating an expected surplus of $\mathbb{E}_\nu[\theta]\coloneqq\int_\Theta \theta \nu(d\theta)$. Therefore, the first-best outcome is the solution to the following surplus-maximization problem:
\begin{align*}
    \label{eq:planner}
    \tag{FB}
    &\max_{\sigma, \tau}\hspace*{.25em}k\int_{\Delta(\Theta)}m(\tau(\nu))\mathbb{E}_\nu[\theta]\sigma_\mu(d\nu) \hspace{1em }\text{s.t. } \hspace{1em }\eqref{eq:feasible}.
\end{align*}

To analyze this surplus-maximization problem, I decompose it into two distinct optimization steps: First, for each market segmentation, I characterize the surplus-maximizing submarket tightness function. Second, I determine the optimal market segmentation. To that end, given market segmentation $\sigma$, the surplus-maximizing submarket tightness function is a solution to 
\begin{align*}
    \label{eq:planner-sigma}
    \tag{FB-$\sigma$}
    &\max_{\tau}\hspace*{.25em}k\int_{\Delta(\Theta)}m(\tau(\nu))\mathbb{E}_\nu[\theta]\sigma_\mu(d\nu) \hspace{1em }\text{s.t. } \hspace{1em }\eqref{eq:feasible}.
\end{align*}

To characterize the solution to \eqref{eq:planner-sigma}, let  the inverse of the mapping $t\mapsto 1/m'(t)$ be denoted by $f$ so that $f(y)=t$ if $1/m'(t)=y$.  Since $m'(t)$ is continuous and strictly decreasing in $t$,  $f(y)$ is continuous and strictly increasing in $y$. Furthermore, under \autoref{ass:meeting-function}, $\lim_{t\to 0} m'(t)=\beta$ and $\lim_{t\to \infty} m'(t)=0$,\footnote{Since $m$ is strictly concave and $m(0)=0$ ( \autoref{ass:meeting-function}),  $m'(t)<m(t)/t$ for all $t>0$. Taking limits on both sides as $t\to 0$, we have $m'(0)\leq \beta$. Similarly, given $t>0$, we have $m'(t)>(m(t')-m(t))/(t'-t)$ for all $t'>t$. Taking double limits on both side, first as $t\to 0$ and then as $t'\to 0$, yields that $m'(0)\geq \beta$.}  which implies that $f:[1/\beta, \infty)\to\mathbb{R}_+$ with  $f(1/\beta)=0$ and $\lim_{y\to \infty}f(y)=\infty$. 
\medskip

\begin{proposition}\label{thm:planner}
Fix a market segmentation $\sigma$. Let $\eta_\sigma$ be the unique value of $\eta \in (0, \beta)$ that solves
\[
k\int_{\Delta(\Theta)}f\left(\frac{\mathbb{E}_\nu[\theta]}{\eta}\right)\cdot\mathbbm{1}_{\left[\beta\,\mathbb{E}_\nu[\theta]>\eta\right]}\sigma_\mu(d\nu)=1.
\]
Then the submarket tightness function $\tau^{FB}_\sigma$ given by 
\[
\tau^{FB}_\sigma(\nu)=\left\{\begin{array}{ccc}0 & \mbox{if} & \beta\,\mathbb{E}_\nu[\theta]\leq \eta_\sigma\\[10pt]
\displaystyle f\left(\frac{\mathbb{E}_\nu[\theta]}{\eta_\sigma}\right) & \mbox{if} & \beta\,\mathbb{E}_\nu[\theta]>\eta_\sigma 
  \end{array}\right.
  \]
  solves \eqref{eq:planner-sigma}. Furthermore, any other submarket tightness function $\tau$ that solves \eqref{eq:planner-sigma} satisfies $\tau=\tau^{FB}$ for $\sigma_\mu$-almost everywhere.
\end{proposition}
\medskip

From a social welfare perspective, thicker submarkets are desirable because they facilitate more meetings, and ultimately, more trade. However, with a finite measure of buyers, increasing the thickness of one submarket necessarily reduces the thickness of others. This trade-off produces two key implications for the first-best allocation of buyers across submarkets.

First, some submarkets may be \emph{inactive}---no buyers are allocated to them. Specifically, given a market segmentation $\sigma$, submarkets generating an expected surplus lower than $\eta_\sigma/\beta$ are inactive, while those generating a strictly larger expected surplus are \emph{active}. In other words, the planner focuses the limited number of buyers on submarkets with a ``high-enough" expected surplus.

Second, the marginal value of an additional buyer in any one of the active submarkets must be equal. As the meeting function is concave, this condition requires that an active submarket's tightness to be strictly increasing in its expected surplus. Furthermore, the marginal value of an additional buyer in any active submarket must exceed that in any inactive submarket. Otherwise, reallocating some buyers from an active to an inactive submarket would increase total surplus. Under this interpretation, any submarket with an expected surplus equal to $\eta_\sigma/\beta$ is a marginal submarket that leaves the planner indifferent between keeping it inactive or allocating buyers to it.

Let us now address the task of characterizing the first-best market segmentation: what market segmentation, along with the corresponding surplus-maximizing submarket tightness function, should the planner choose? In other words, the first-best market segmentation is the solution to the following:
\begin{align*}
    &\max_{\sigma}\hspace*{.25em}k\int_{\Delta(\Theta)}m(\tau^{FB}_\sigma(\nu))\mathbb{E}_\nu[\theta]\sigma_\mu(d\nu).
\end{align*}

To characterize the first-best market segmentation, I begin by defining the notion of a \emph{perfect market segmentation}, which allocates each seller type to a distinct submarket. Formally, for each $\theta \in \Theta$, let $\delta_\theta \in \Delta(\Theta)$ denote the Dirac measure at $\theta$, so that for any $A \in \mathcal{B}(\Theta)$, $\delta_\theta(A) = \mathbbm{1}_{[\theta\in A]}$. The perfect market segmentation is then defined by the stochastic kernel $\sigma^{PS}$ such that, for all $B \in \mathcal{B}(\Delta(\Theta))$ and $\theta \in \Theta$, $\sigma^{PS}(B,\theta)=\mathbbm{1}_{[\delta_\theta\in B]}$.

\begin{proposition}\label{thm:first}
The perfect segmentation $\sigma^{PS}$ is a first-best market segmentation.
\end{proposition}\medskip

Intuitively, by choosing a market segmentation $\sigma$ and a corresponding submarket tightness function $\tau_\sigma^{FB}$, the planner is effectively allocating
\[
k\mu(d\theta)\int_{\Delta(\Theta)}m(\tau_\sigma^{FB}(\nu))\sigma(d\nu, \theta)
\]
meetings to type-$\theta$ sellers. The concavity in the meeting function then implies that,  for each seller type, concentrating buyers into a single thick submarket  minimizes the mass of buyers needed to attain a desired volume of meetings. Consequently, perfect segmentation emerges as a first-best market segmentation.

Nevertheless, perfect segmentation is not the only first-best market segmentation. For instance, under perfect segmentation, we know from \autoref{thm:planner} that the submarket for any seller type below $\eta_{\sigma^{PS}}/\beta$ is inactive, while the submarket for any seller type strictly above $\eta_{\sigma^{PS}}/\beta$ is active. Thus, a lower censorship market segmentation, which pools all seller types $\theta\leq\eta_{\sigma^{PS}}/\beta$ into a single submarket and maintains a perfect segmentation for all types $\theta> \eta_{\sigma^{PS}}/\beta$, would also constitute a first-best market segmentation. 

\section{Search Equilibrium}\label{sec:equilibrium}
Let us now consider a search market in which each buyer strategically chooses which submarkets to enter. What is the  equilibrium  outcome of such a decentralized search environment? 

To address this, consider an arbitrary market segmentation $\sigma$, and suppose buyers anticipate a submarket tightness function $\tau:\Delta(\Theta)\to\mathbb{R}_+$. Given $(\sigma, \tau)$, a buyer who joins submarket $\nu\in\supp(\sigma_\mu)$ expects to meet a seller with probability $m(\tau(\nu))/\tau(\nu)$. If the buyer meets a type-$\theta$ seller within the submarket (with $\theta\in\supp(\nu)$), he further expects to earn a payoff $\lambda(\theta)\theta$ from trade. Therefore, the buyer's expected payoff from joining submarket $\nu$ is given by 
\[
U(\nu;\tau)\coloneqq  \frac{m
\big(\tau(\nu)\big)}{\tau(\nu)}\underbrace{\int_\Theta \lambda(\theta)\theta\nu(d\theta)}_{\coloneqq\mathbb{E}_\nu[\lambda(\theta)\theta]}.
\]

\begin{definition}\label{def:equilibrium}  
A search equilibrium of a market segmentation $\sigma$ is given by a submarket tightness function $\tau:\Delta(\Theta)\to\mathbb{R}_+$ and a reservation value $u\in \mathbb{R}$ such that
\begin{enumerate}[$(a)$]
    \item $\tau$ satisfies the feasibility constraint \eqref{eq:feasible}, and
    \item $(\tau, u)$ jointly satisfy 
\[
\label{eq:inequality}
\tag{3}
U(\nu;\tau)\leq u
\]
for $\sigma_\mu$-almost all $\nu$ with equality if $\tau(\nu)>0$.
\end{enumerate}
\end{definition} 

In essence, buyers form expectations about submarket tightness $\tau$ and their reservation value $u$, which represents the average payoff they expect to earn from participating in the search market. In equilibrium, these expectations must be consistent with a feasible allocation of buyers across submarkets, as captured by \eqref{eq:feasible}, as well as with an optimal buyers' search strategy, as captured by \eqref{eq:inequality}. Intuitively, the latter condition states that submarkets that yield a payoff lower than $u$ remain inactive, while those offering payoffs at least as large as $u$ are active. Furthermore, among active submarkets, buyers must be indifferent between any two; otherwise, they would all gravitate toward the submarket offering the higher expected payoff. 

Finally, a search equilibrium $(\tau, u)$ must satisfy a rational expectations condition: the anticipated reservation value $u$ should be consistent with the buyers' actual ex-ante payoff. To see this, observe that a pair $(\tau, u)$ that satisfies  \autoref{def:equilibrium} implies that $m(\tau(\nu))\mathbb{E}_\nu[\lambda(\theta)\theta]=u\,\tau(\nu)$ for $\sigma_\mu$-almost all $\nu$ from \eqref{eq:inequality}. Integrating over all submarkets then yields 
\begin{align*}
\int_{\Delta(\Theta)}m(\tau(\nu))\mathbb{E}_\nu[\lambda(\theta)\theta]\sigma_\mu(d\nu)&=u\int_{\Delta(\Theta)}\tau(\nu)\sigma_\mu(d\nu)\\[6pt]
    &=\frac{u}{k},
\end{align*}
where the last equality follows from \eqref{eq:feasible}. On the other hand, buyers' ex-ante payoff from an allocation policy $Q\in \Delta(\Delta(\Theta))$ with $dQ/d\sigma_\mu=k\cdot\tau$ is given by\footnote{Recall that given segmentation $\sigma$, submarket $\nu\in \supp(\sigma_\mu)$ has a measure $k \cdot\sigma_\mu(d\nu)$ sellers. Hence, if the submarket has tightness $\tau(d\nu)$, then it must have a measure $Q(d\nu)$ of buyers, where $Q(d\nu)=k\cdot\tau(d\nu)\cdot \sigma_\mu(d\nu)$.} 
\begin{align*}
    \int_{\Delta(\Theta)}U(\nu;\tau)Q(d\nu)&= k\int_{\Delta(\Theta)}m(\tau(\nu))\mathbb{E}_\nu[\lambda(\theta)\theta]\sigma_\mu(d\nu).
\end{align*}
Putting these two expressions together then yields the desired rational expectations condition: 
\begin{align*}
\label{eq:equil}
\tag{4}
u=k\int_{\Delta(\Theta)}m(\tau(\nu))\mathbb{E}_\nu[\lambda(\theta)\theta] \sigma_\mu(d\nu),
\end{align*}
that is, the buyers' reservation value is equal to the share of the ex-ante surplus they capture in equilibrium. 

To characterize the equilibrium, let the inverse of the mapping $t\mapsto t/m(t)$ be denoted by $g$ so that $g(y)=t$ if $t/m(t)=y$.  Since $m(t)/t$ is continuous and strictly decreasing in $t$, $g(y)$ is continuous and strictly increasing in $y$. Furthermore, $m(0)/0=\beta$ and $\lim_{t\to \infty} m(t)/t=0$, which implies that  $g:[1/\beta, \infty)\to\mathbb{R}_+$ with  $g(1/\beta)=0$ and $\lim_{y\to \infty}g(y)=\infty$. 
\medskip

\begin{proposition}\label{thm:equil}
Fix a market segmentation $\sigma$. There exists a search equilibrium $(\tau^*_\sigma, u^*_\sigma)$ such that $0< u^*_\sigma<\beta\sup_{\theta\in\Theta}\lambda(\theta)\theta$, and
\[
\tau^*_\sigma(\nu)=\left\{\begin{array}{ccc}0 & \mbox{if} & \beta\,\mathbb{E}_\nu[\lambda(\theta)\theta]\leq u^*_\sigma\\[4pt]
\displaystyle g\left(\frac{\mathbb{E}_\nu[\lambda(\theta)\theta]}{u^*_\sigma}\right) & \mbox{if} & \beta\,\mathbb{E}_\nu[\lambda(\theta)\theta]>u^*_\sigma 
  \end{array}\right..
  \]
  Furthermore, the equilibrium is essentially unique: any other search equilibrium $(\tau, u)$ satisfies $\tau=\tau^*_\sigma$ for $\sigma_\mu$-almost everywhere, and $u=u^*_\sigma$.
\end{proposition}
\medskip

To build intuition, consider a fixed reservation value $u$. When buyers anticipate a payoff of $u$, they optimally avoid entering any submarket $\nu$ satisfying $\beta\cdot \mathbb{E}_\nu[\lambda(\theta)\theta]\leq u$, even if these submarkets offer the highest probability of meeting a seller. This is because the expected payoff from trade in such submarkets is too low to justify participation. Among the remaining submarkets, those offering a higher payoff conditional on trade attract more buyers, leading to greater congestion. As a result, congestion in each active submarket precisely offsets its potential gains from trade. Consequently, the optimal submarket tightness for a given reservation value $u$ is determined by the mapping 
\[
\nu\mapsto g\left(\frac{\mathbb{E}_\nu[\lambda(\theta)\theta]}{u}\right)\mathbbm{1}_{[\beta\, \mathbb{E}_\nu[\lambda(\theta)\theta]>u]}.
\]
Moreover, a buyer's average payoff across all active submarkets is given by
\[
\int_{\Delta(\Theta)}m\left(g\left(\frac{\mathbb{E}_\nu[\lambda(\theta)\theta]}{u}\right)\right)\mathbbm{1}_{[\beta\, \mathbb{E}_\nu[\lambda(\theta)\theta]>u]}\mathbb{E}_\nu[\lambda(\theta)\theta]\sigma_\mu(d\nu),
\]
which is strictly decreasing in $u$. In other words, when all buyers anticipate a higher reservation value, fewer submarkets remain active, increasing congestion within those submarkets. This, in turn, lowers the probability that any given buyer meets a seller, thereby reducing their average payoff. Of course, in equilibrium, the reservation value $u^*_\sigma$ is not arbitrary; rather, it is determined by the rational expectations condition, captured in \eqref{eq:equil}, which requires it to equal the buyer's average payoff across all active submarkets.

A comparison of the socially efficient allocation of buyers across submarkets (\autoref{thm:planner}) to the equilibrium allocation  (\autoref{thm:equil}) highlights two main sources of inefficiency. First, a social planner's preference over seller types is strictly monotonic, as a higher type generates a larger surplus. In contrast, a buyer's valuation of seller types depends on the share of surplus they capture, which may be non-monotonic in seller types depending on the surplus-splitting function. For example, if $\lambda(\theta) = 1-\theta$, buyers prefer to meet intermediate-type sellers over low- and high-type sellers. This divergence in preferences over seller types can lead to a misalignment in how the planner and buyers evaluate submarkets. Consequently, the set of active submarkets in a search equilibrium may differ from that of the socially efficient allocation.

Second, even when the planner's and buyers' preferences over seller types align, the allocation of buyers across these submarkets may differ. In equilibrium, buyers allocate themselves to equalize their expected payoffs across active submarkets. In contrast, the planner allocates buyers to equalize the marginal value of an additional meeting across active submarkets. The equilibrium allocation is driven by the level of congestion in a submarket, $m(t)/t$, while the planner's allocation is driven by the thick market externality, $m'(t)$. From \autoref{ass:meeting-function}, $m(t)/t\geq m'(t)$ for all $t\geq 0$, implying that, even when a submarket is active under both the socially efficient and equilibrium allocations, too many buyers join the submarket in equilibrium, resulting in inefficiently high levels of congestion. 

\section{Efficiency of Search Equilibrium}\label{sec:equil-efficiency}

\subsection{Hosios condition}\label{sec:hosios}
Despite the potential inefficiencies in the equilibrium allocation of buyers, this section identifies necessary and sufficient conditions under which the first-best outcome, as established in \autoref{thm:planner} and \autoref{thm:first}, can be decentralized. To state the result, let 
 \[
 \varepsilon(t)\coloneqq \frac{m'(t)}{m(t)}\, t
 \]
denote the elasticity of the meeting function for $t\geq 0$.

\begin{proposition}\label{thm:hosios}
Suppose the surplus-splitting function $\lambda:\Theta\to [0,1]$ is continuous. The pair $(\sigma^{PS}, \tau^*_{\sigma^{PS}})$ solves \eqref{eq:planner} if and only if for all $\theta\leq\eta_{\sigma^{PS}}/\beta$, 
    \[
    \lambda(\theta)\theta\leq \lambda\left(\frac{\eta_{\sigma^{PS}}}{\beta}\right) \frac{\eta_{\sigma^{PS}}}{\beta},
    \]
    and for all $\theta>\eta_{\sigma^{PS}}/\beta$ 
    \[
\lambda(\theta)=\lambda\left(\frac{\eta_{\sigma^{PS}}}{\beta}\right)\varepsilon\left(f\left(\frac{\theta}{\eta_{\sigma^{PS}}}\right)\right).
\]
\end{proposition}
\medskip

The proposition relates the efficiency of decentralized search to the elasticity of the meeting function and the buyers' bargaining power, giving rise to a \emph{Hosios condition} \citep{hos90}. Intuitively, the condition ensures that the surplus captured by a buyer within any inactive submarket is no greater than what he could capture from joining the submarket of a type-$\eta_{\sigma^{PS}}/\beta$ seller (the marginal submarket in the first-best outcome). Moreover, a buyer's expected payoff from joining any active submarket is exactly balanced against the marginal value of an additional buyer in that submarket. This balance implies that buyers internalize the congestion and thick market externalities they create when choosing a submarket, aligning their private incentives with the social planner’s objectives and leading to an efficient outcome.

An instructive special case that satisfies both conditions of \autoref{thm:hosios} is when $\lambda(\theta)=0$ for all $\theta\in\Theta$. In this case, buyers would be completely indifferent across all submarkets since they capture no surplus, eliminating any tension between private and social incentives. Clearly, efficiency is straightforward to attain in equilibrium under these conditions. However, this trivial case is ruled out by \autoref{ass:2}, which ensures that buyers capture a strictly positive share of surplus.

\subsection{Constrained-efficient market segmentation}\label{sec:constrained-efficient}
A key takeaway from \autoref{thm:hosios} is the restrictive nature of the Hosios condition. Achieving efficiency in a decentralized search equilibrium often requires a knife-edge condition that rarely holds in many familiar settings. For instance, inefficiency persists when the surplus-splitting function is constant and positive, which is a natural assumption if seller types determine surplus without directly affecting bargaining power. This rarity of efficiency in equilibrium  raises an important question: what is the constrained-efficient segmentation of a search market? That is, how should the social planner optimally segment the market when it can no longer directly control the allocation of buyers across submarkets but can still influence their choices indirectly through the information they observe prior to selecting a submarket?

Formally, each market segmentation $\sigma$ induces an essentially-unique search equilibrium $(\tau^*_\sigma, u^*_\sigma)$, as characterized by \autoref{thm:equil}. A planner who implements a segmentation $\sigma$ expects to create a measure $k\sigma_\mu(d\nu)m(\tau_\sigma^*(\nu))$ of meetings in submarket $\nu\in\supp(\sigma_\mu)$, with each meeting generating an expected surplus of $\mathbb{E}_\nu[\theta]$. Thus, the constrained-efficient market segmentation is the solution to the planner's second-best problem: 
\[
\label{eq:sb}
\tag{SB}
\max_{\sigma}\hspace*{.5em}k\int_{\Delta(\Theta)}m(\tau^*_\sigma(\nu))\mathbb{E}_\nu[\theta]\sigma_\mu(d\nu).
\]

This optimization problem is inherently complex: it involves optimizing over an infinite-dimensional choice variable that enters the objective function in a highly non-linear way. Specifically, the choice of segmentation influences buyers' beliefs about seller types within each submarket, which in turn affects how buyers allocate themselves across submarkets.\footnote{Although \eqref{eq:sb} resembles an information design problem, it differs in a crucial way: whereas a receiver's optimal action in standard information design depends solely on the realized posterior beliefs, buyers' search strategies here depend both on the entire distribution of posterior beliefs.}  In \Cref{necessary}, I derive a sufficient and necessary condition  that decouples these two interdependent effects. However, even this requires solving a constrained concavification problem \citep{kam11} with a continuum of  states. As such,  a general characterization of the constrained-efficient market segmentation appears intractable in many cases.

To make further progress, I introduce the following additional assumption:\medskip

\begin{assumption}
    \label{ass:3}
    For all $\theta\in \Theta$, $\lambda(\theta) = \ell$ for some $\ell\in (0,1]$.
\end{assumption}\medskip

Under this assumption, buyers capture a fixed share $\ell$ of surplus generated in any submarket. Hence, maximizing total surplus is now equivalent to maximizing buyers' welfare. Moreover, \autoref{ass:3} implies that, in equilibrium, any two submarkets with identical posterior means have the same submarket tightness and generate the same surplus, i.e., for $\nu,\nu'\in \supp(\sigma_\mu)$ with $\mathbb{E}_\nu[\theta]=\mathbb{E}_{\nu'}[\theta]$, we have $\tau_\sigma^*(\nu)=\tau^*_\sigma(\nu')$, and $m(\tau^*_\sigma(\nu))\mathbb{E}_\nu[\theta]=m(\tau^*_\sigma(\nu'))\mathbb{E}_{\nu'}[\theta]$. This property allows us to reformulate \eqref{eq:sb} as an optimization problem over posterior-mean distributions with a majorization constraint. 

To that end, let $F$ denote the cumulative distribution function (CDF) associated with the prior $\mu\in\Delta(\Theta)$. Importantly, $F$ is also the posterior-mean distribution induced by the perfect market segmentation.\footnote{Formally, $
F(x)\coloneqq\sigma^{PS}_\mu\big(\{\nu\in \Delta(\Theta):\mathbb{E}_\nu[\theta]\leq x\}\big)$
for all $x\in\Theta$.} More generally, a CDF $H:\Theta\to [0,1]$ represents a posterior-mean distribution induced by some market segmentation if and only if $H$ is a mean-preserving contraction of $F$ \citep{bla53, blackwell_girschik54}. Let $\mpc(F)$ denote the set of all such mean-preserving contractions.

For each $H\in\mpc(F)$, there exists some market segmentation $\sigma$ that induces $H$ as its posterior-mean distribution, and yields an essentially-unique search equilibrium $(\tau^*_\sigma, u^*_\sigma)$. With some abuse of language, I refer to $H$ as the market segmentation, and denote the search equilibrium as $(\tau_H^*, u_H^*)$ where 
$\tau^*_H(\mathbb{E}_\nu[\theta])\coloneqq\tau^*_\sigma(\nu)$ and $u^*_H\coloneqq u^*_\sigma$.\footnote{Note that the submarket tightness function is now defined over $\Theta$, i.e., $\tau_H^*:\Theta\to\mathbb{R}_+$ maps a posterior mean $x\in\Theta$ to submarket tightness $\tau^*_H(x)$.} 

The planner's second-best problem \eqref{eq:sb} can now be reformulated as follows:
\[
\label{eq:sb2}
\tag{SB$'$}
\max_{H\in\mpc(F)}\hspace*{.5em}k\int_{\Theta}m(\tau^*_{H}(x))x d H(x). 
\] 
While this optimization problem remains non-linear, the next result provides a tractable characterization of a constrained-efficient market segmentation. 

\begin{proposition}
\label{thm:secondbest}
 A segmentation $H\in \mpc(F)$ is constrained-efficient if and only if 
    \[
    \label{eq:simplify}
    \tag{5}
H\in\argmax_{\widehat H\in\mpc(F)}\hspace*{.5em}k\int_{\Theta}\tau^*_{H}(x) d\widehat H(x). 
    \]
Furthermore, 
\begin{enumerate}[$(\roman*)$]
    \item If $t\mapsto t/m(t)$ is concave, then $F$ is a constrained-efficient segmentation.
    \item  If $t\mapsto t/m(t)$ is convex, then there exists a cutoff  $\theta^c\in\Theta$, and posterior means $\underline x\coloneqq \mathbb{E}_{F}[\theta|\theta\leq \theta^c]$ and  $\bar x\coloneqq \mathbb{E}_{F}[\theta|\theta\geq \theta^c]$ such that
   \[
   H(x)=\left\{\begin{array}{ccc}
0&\mbox{if}        &  x<\underline x\\
F(\theta^c)&\mbox{if}        &  x\in [\underline x, \bar x)\\
1&\mbox{if}        &  x\geq \bar x\\
   \end{array}\right.
   \]
    is a constrained-efficient segmentation. In this case, $\tau_H^*(\underline x)=0<\tau_H^*(\bar x)$. 
\end{enumerate}
\end{proposition}
\medskip

The above characterization has two implications for the constrained-efficient segmentation problem. First, it simplifies the non-linear optimization problem in \eqref{eq:sb2} by decomposing it into a linear programming problem and a fixed point problem involving the feasibility constraint. The linear programming problem can be addressed by applying tools from the information design literature \citep{dwo19, kle21, kol22, ari23}. Specifically, the optimal segmentation can be characterized by a monotone partition, where the type space is divided into a countable number of connected intervals. Within each interval, the planner allocates each seller type to her own distinct submarket (perfect revelation); or allocates all seller types to same submarket (pooling); or randomly allocates all seller types to one of two submarkets (bi-pooling).

Second, under additional assumptions on $t/m(t)$, which represents the buyers' \emph{odds} of meeting a seller (a 1-in-$t/m(t)$ chance), the constrained-efficient market segmentation takes one of two forms: perfect or binary market segmentation. To build intuition, recall that efficiency requires allocating more meetings to higher types. While the planner cannot directly control the number of meetings each seller type receives, it can indirectly influence how buyers distribute themselves across submarkets through the market segmentation. In particular, the planner can separate higher and lower seller types into distinct submarkets, which induces more buyers to self-select into the submarket for higher seller types, thereby increasing the number of meetings for those sellers

However, this strategy comes with trade-offs. While separating seller types ensures high-type sellers receive more meetings, the concavity of the meeting function implies a reduction in the total number of meetings in the search market. For this segmentation to enhance welfare, the surplus gained from increasing meetings for high-type sellers must outweigh the loss of surplus stemming from a lower number of overall meetings.

\autoref{thm:secondbest} demonstrates that such segmentation is welfare-improving when $t/m(t)$ is concave. Intuitively, as $t$ increases, a buyer's odds of meeting a seller worsen, reflecting the congestion externality buyers face in tighter markets. When $t/m(t)$ is concave, this congestion externality worsens gradually as market tightness increases. Consequently, separating high and low seller types into distinct submarkets encourages many buyers to allocate themselves to the submarket for high-type sellers, mitigating the loss in welfare from a decrease in the total number of meetings. Consequently, perfect segmentation emerges as constrained-efficient. Conversely,
when $t/m(t)$ is convex, the congestion externality worsens more rapidly with increased submarket tightness, making such segmentation less attractive for the planner.

Let us conclude by noting that many meeting functions widely used in the search and matching literature satisfy the concavity or convexity conditions identified in \autoref{thm:secondbest}. For example, the mapping $t\mapsto t/m(t)$ is concave for CES meeting functions with parameter values $\rho\leq 1$, leading to perfect segmentation as the constrained-efficient outcome. In this case, buyers leverage granular ex-ante information about seller types to guide their choice of submarkets, thereby endogenously creating a directed search market where buyers learn a seller's type before meeting her.

In contrast, for Urn-ball meeting functions and for CES meeting functions with parameter values $\rho\geq 1$, the mapping $t\mapsto t/m(t)$ is convex. Similarly, while the generalized frictionless meeting function with $m(t)=\min\{\alpha,\beta t\}$ does not satisfy all the conditions of \autoref{ass:meeting-function}, it inherits the properties of a CES meeting function, which approximates the frictionless meeting function as $\rho\to \infty$. As a result, the constrained-efficient segmentation in all three cases is a binary segmentation. In these cases, buyers have minimal ex-ante information they can use to direct their search. Instead, a random search market emerges, with buyers learning a seller's type only after meeting her.

\newpage
\onehalfspacing
\appendix
\section{Appendix}\label{appendix}
\subsection{Proofs}\label{proofs}
\begin{proof}[Proof of \autoref{thm:planner}]
Let us fix a market segmentation $\sigma$. Let $\mathcal{L}^1(\Delta(\Theta), \sigma_\mu)$ be the space of  $\sigma_\mu$-integrable functions from $\Delta(\Theta)$ to $\mathbb{R}$. Consider the following relaxed problem:
\begin{align*}
    \max_{\substack{\tau \in \mathcal{L}^1(\Delta(\Theta), \sigma_\mu)\\[2pt]\tau\geq 0}}\hspace*{.25em} &k\int_{\Delta(\Theta)}m(\tau(\nu))\mathbb{E}_\nu[\theta]\sigma_\mu(d\nu)\\[6pt]
    \text{s.t. } & \hspace*{.25em}k\int_{\Delta(\Theta)}\tau (\nu)\sigma_\mu(d\nu)\leq 1.
\end{align*}
Notice that the constraint in the relaxed problem must bind at the optimum since $m(t)$ is a strictly increasing function. Hence, any solution to the relaxed problem is also a solution to \eqref{eq:planner-sigma}.

Notice that both the objective and the constraint in the relaxed problem are concave  functionals. Thus, from \cite{luenberger} (Section 8.3, Theorem 1; Section 8.4, Theorem 2), $\tau_\sigma^{FB}$ is a solution to \eqref{eq:planner-sigma} if and only if there exists a Lagrange multiplier $\eta_\sigma> 0$ such that $\tau_\sigma^{FB}$ solves: 
\[
\max_{\tau\in \mathcal{L}^1(\Delta(\Theta), \sigma_\mu) }L(\tau,\eta_\sigma)\coloneqq k\int_{\Delta(\Theta)}\Big[m(\tau(\nu))\mathbb{E}_\nu[\theta]-\eta_\sigma\tau(\nu)\Big]\sigma_\mu(d\nu)+\eta_\sigma.
\]

We can maximize the Lagrangian by pointwise maximizing the integrand, which depends on $\sigma$ only indirectly through $\eta_\sigma$. Hence, we can first pointwise maximize $L(\tau, \eta)$ for an arbitrary $\eta>0$, and then find the value $\eta_\sigma$ that makes the constraint in \eqref{eq:planner-sigma} bind.

To that end, given $\eta>0$, let us find a function $\psi:\Delta(\Theta)\times(0,\infty)\to\mathbb{R}_+$ such that $m'(\psi(\nu, \eta))\mathbb{E}_\nu[\theta]\leq \eta$ for all $\nu\in\supp(\sigma_\mu)$ with  equality if $\psi(\nu, \eta)>0$. These first-order conditions are also sufficient as $m(t)\mathbb{E}_\nu[\theta]-\eta t$ is concave in $t$ for all $\nu\in\Delta(\Theta)$. Recall that $m'(t)$ is strictly decreasing in $t$ with $m'(0)=\beta$. Thus, set $\psi(\nu, \eta)=0$ for all $\nu\in\supp(\sigma_\mu)$ with $\beta\, \mathbb{E}_\nu[\theta]\leq \eta$. Furthermore, for all $\nu\in\supp(\sigma_\mu)$ with $\beta\,\mathbb{E}_\nu[\theta]>\eta$, set $m'(\psi(\nu, \eta))\mathbb{E}_\nu[\theta]=\eta$. As a result, recalling the definition of $f:[1/\beta, \infty)\to \mathbb{R}_+$, we have
\[
\psi(\nu, \eta)\coloneqq\left\{\begin{array}{ccc}
 0    &  \mbox{if} & \beta\,\mathbb{E}_\nu[\theta]\leq \eta\\
f\left(\displaystyle\frac{\mathbb{E}_\nu[\theta]}{\eta}\right) &  \mbox{if} & \beta\,\mathbb{E}_\nu[\theta]> \eta
\end{array}\right..
\]

Since $f$ is a continuous and strictly increasing function with $\lim_{y\to \infty}f(y)=\infty$ and $f(1/\beta)=0$, we have:
\begin{enumerate}[$(i)$]
    \item the mapping $(\nu, \eta)\mapsto \psi(\nu, \eta)$ is continuous in both arguments,
    \item for each $\eta>0$, $\psi(\nu, \eta)\geq \psi(\nu',\eta)$ whenever $\mathbb{E}_\nu[\theta]\geq \mathbb{E}_{\nu'}[\theta]$, 
    \item for each $\nu\in\supp(\sigma_\mu)$, $\psi(\nu, \eta)\leq \psi(\nu,\eta')$ whenever $\eta> \eta'>0$, and 
    \item for $\sigma_\mu$-almost all $\nu$, $\lim_{\eta\to 0}\psi(\nu, \eta)=\infty$, and $\psi(\nu, \eta)=0$ for all $\eta\geq \beta$.
\end{enumerate}
Consequently, $\psi(\nu, \eta)\leq \psi(\delta_1, \eta)<\infty$ for all $\nu\in\supp(\sigma_\mu)$ and all $\eta>0$, which in turn implies $\psi(\cdot, \eta)\in \mathcal{L}^1(\Delta(\Theta), \sigma_\mu)$ for all $\eta>0$.

Next, let us show that there exists a unique multiplier $\eta_\sigma\in (0, \beta)$ such that $\eta_\sigma$ along with $\psi(\cdot,\eta_\sigma)$  lead to a binding feasibility constraint in \eqref{eq:planner-sigma}. To that end, define $\Psi:(0, \infty)\to \mathbb{R}_+$ to be the mapping given by 
\[
\Psi(\eta)\coloneqq k\int_{\Delta(\Theta)}\psi(\nu, \eta)\sigma_\mu(d\nu).
\]
Clearly, the  feasibility constraint in \eqref{eq:planner-sigma} binds at $\eta_\sigma$ if $\Psi(\eta_\sigma)=1$. 

Observe that $\Psi$ is a continuous and weakly decreasing function because $\psi(\nu, \eta)$ is continuous and weakly decreasing in $\eta$ for each $\nu\in\supp(\sigma_\mu)$. Moreover, $\Psi(\eta)=0$ for all $\eta\geq \beta$, and from the monotone convergence theorem, 
\[
\lim_{\eta\to 0}\Psi(\eta)=k\int_{\Delta(\Theta)} \lim_{\eta\to 0} \psi(\nu, \eta)\sigma_\mu(d\nu)=\infty.
\]
Thus, there exists a unique $\eta_\sigma\in (0, \beta)$ such that $\Psi(\eta_\sigma)=1$. The optimal submarket tightness is defined by $\tau_\sigma^{FB}\coloneqq \psi(\cdot, \eta_\sigma)$. Furthermore, any $\tau$ that solves \eqref{eq:planner-sigma} must have $\tau=\tau^{FB}_\sigma$ for $\sigma_\mu$-almost everywhere. In other words, the optimal submarket tightness for a given segmentation $\sigma$ is unique up to the  $\sigma_\mu$-null sets. 
\end{proof}
\medskip

\begin{proof}[Proof of \autoref{thm:first}]
Fix any market segmentation $\sigma$, and let $\tau^{FB}_\sigma$ be the solution to \eqref{eq:planner-sigma}. The proof proceeds by showing that, given the pair $(\sigma, \tau_\sigma^{FB})$, there exists another pair $(\sigma^{PS}, \tilde\tau)$ that achieves a (weakly) higher surplus for each type-$\theta$ seller. To that end, define a new submarket tightness function $\tilde\tau:\Delta(\Theta)\to \mathbb{R}_+$ given by 
\[
\tilde\tau(\nu)=\left\{\begin{array}{ccc}
\int_{\Delta(\Theta)} \tau_\sigma^{FB}(\nu') \sigma(d\nu', \theta)& \mbox{if}    &\nu=\delta_\theta  \\
 \tau_\sigma^{FB}(\nu)&    & \mbox{otherwise}
\end{array}\right.,
\]
and notice that
\[
\int_{\Delta(\Theta)}\tilde \tau(\nu)\sigma_\mu^{PS}(d\nu)=\int_\Theta \tilde \tau(\delta_\theta)\mu(d\theta)=\int_{\Delta(\Theta)} \tau^{FB}_\sigma(\nu) \sigma_\mu(d\nu),
\]
where the first equality follows from the definition of $\sigma^{PS}$ and the second equality follows by construction of $\tilde \tau$. Therefore, $\tilde\tau\in\mathcal{L}^1(\Delta(\Theta), \sigma^{PS}_\mu)$ and the pair
$(\sigma^{PS}, \tilde\tau)$ satisfies \eqref{eq:feasible}. 

Comparing the expected surplus generated from a type-$\theta$ seller under $(\sigma, \tau^{FB}_\sigma)$ to that under $(\sigma^{PS}, \tilde \tau)$, we have
\begin{align*}
k\cdot \theta\int_{\Delta(\Theta)} m(\tau^{FB}_\sigma(\nu))\sigma(d\nu, \theta)& \leq  k\cdot \theta \cdot m\left(\int_{\Delta(\Theta)}\tau^{FB}_\sigma(\nu)\sigma(d\nu,\theta)\right)\\[6pt]
    &= k\cdot \theta \cdot m(\tilde\tau(\delta_\theta))\\[6pt]
    &=k\cdot \theta\int_{\Delta(\Theta)} m(\tilde \tau(\nu))\sigma^{PS}(d\nu, \theta),
\end{align*}
where the inequality follows from the concavity of $m$, the first equality follows by definition of $\tilde\tau$, and the last follows by the definition of $\sigma^{PS}$. 

Thus, for any arbitrary market segmentation $\sigma$, the pair $(\sigma, \tau_\sigma^{FB})$ generates a (weakly) lower expected surplus than the pair $(\sigma^{PS}, \tilde\tau)$. Of course, $(\sigma^{PS}, \tilde\tau)$ generates a (weakly) lower expected surplus than the pair $(\sigma^{PS}, \tau^{FB}_{\sigma^{PS}})$, because the latter solves \eqref{eq:planner-sigma} when $\sigma=\sigma^{PS}$. As a result, the pair $(\sigma^{PS}, \tau^{FB}_{\sigma^{PS}})$ is a solution to \eqref{eq:planner}.
\end{proof}\medskip

\begin{proof}[Proof of \autoref{thm:equil}] Let us first establish a sufficient and necessary condition for a submarket to be active. We shall then use this result to show that an essentially-unique equilibrium exists for any given market segmentation.

\begin{lemma}
\label{lemma:active markets}
Suppose $(\tau^*_\sigma, u^*_\sigma)$ is a search equilibrium for a market segmentation $\sigma$. For $\sigma_\mu$-almost all $\nu$, $\tau^*_\sigma(\nu)>0$ if and only if $\beta\,\mathbb{E}_\nu[\lambda(\theta)\theta]>u^*_\sigma$.
\end{lemma}
\begin{proof}[Proof of \autoref{lemma:active markets}]\

(\emph{Only-if direction}): Suppose submarket $\nu\in \supp(\sigma_\mu)$ is active, i.e., $\tau^*_\sigma(\nu)>0$, which implies 
\[
u^*_\sigma=\frac{m(\tau^*_\sigma(\nu))}{\tau^*_\sigma(\nu)}\mathbb{E}_\nu[\lambda(\theta)\theta]< \beta\,\mathbb{E}_\nu[\lambda(\theta)\theta],
\]
where the equality follows from \eqref{eq:inequality}, and the inequality follows from the fact $\beta$ is the highest probability with which a buyer meets a seller.
\medskip

(\emph{If direction}): Suppose submarket $\nu\in \supp(\sigma_\mu)$ is inactive, i.e., $\tau^*_\sigma(\nu)=0$, which implies 
\[
u^*_\sigma\geq\frac{m(0)}{0}\mathbb{E}_\nu[\lambda(\theta)\theta]= \beta\,\mathbb{E}_\nu[\lambda(\theta)\theta],
\]
where the inequality follows from \eqref{eq:inequality}, and the equality follows by definition that $m(0)/0=\beta$. Thus, from the contrapositive,  $\beta\,\mathbb{E}_\nu[\lambda(\theta)\theta]>u^*_\sigma$ implies $\tau^*_\sigma(\nu)>0$. 
\end{proof}
\medskip

Let us now turn to proving the existence of an essentially-unique equilibrium. By \autoref{lemma:active markets}, any search equilibrium $(\tau^*_\sigma, u^*_\sigma)$ satisfies $\tau^*_\sigma(\nu)=0$ for $\sigma_\mu$-almost all $\nu$ such that $ \beta\,\mathbb{E}_\nu[\lambda(\theta)\theta]\leq u^*_\sigma$, and  $\tau^*_\sigma(\nu)>0$ for $\sigma_\mu$-almost all $\nu$ such that $\beta \,\mathbb{E}_\nu[\lambda(\theta)\theta]>u^*_\sigma$. In the latter case, \eqref{eq:inequality} implies that 
\[
\frac{\tau^*_\sigma(\nu)}{m(\tau^*_\sigma(\nu))}=\frac{\mathbb{E}_\nu[\lambda(\theta)\theta]}{u^*_\sigma},
\]
which from the definition of $g$ implies that $\tau^*_\sigma(\nu)=g\left(\mathbb{E}_\nu[\lambda(\theta)\theta]/u^*_\sigma\right)$. In other words, there exists a mapping $\phi:\Delta(\Theta)\times (0,\infty)\to \mathbb{R}_+$ given by
\[
\phi(\nu, u)\coloneqq \left\{\begin{array}{ccc}0 & \mbox{if} & \beta\,\mathbb{E}_\nu[\lambda(\theta)\theta]\leq u\\[4pt]
\displaystyle g\left(\frac{\mathbb{E}_\nu[\lambda(\theta)\theta]}{u}\right) & \mbox{if} & \beta\,\mathbb{E}_\nu[\lambda(\theta)\theta]>u
  \end{array}\right.
  \]
such that $\tau^*_\sigma=\phi(\cdot, u^*_\sigma)$ for $\sigma_\mu$-almost everywhere. 

Since $g$ is a continuous and  strictly increasing function with $\lim_{y\to \infty}g(y)=\infty$ and $g(1/\beta)=0$, we have:
\begin{enumerate}[$(i)$]
    \item the mapping $(\nu, u)\mapsto \phi(\nu, u)$ is continuous in both arguments,
    \item for each $u>0$, $\phi(\nu, u)\geq \psi(\nu',u)$ whenever $\mathbb{E}_\nu[\lambda(\theta)\theta]\geq \mathbb{E}_{\nu'}[\lambda(\theta)\theta]$, and
    \item for each $\nu\in\supp(\sigma_\mu)$, $\phi(\nu, u)\leq \phi(\nu,u')$ whenever $u>u'>0$, and
    \item for $\sigma_\mu$-almost all $\nu$, $\lim_{u\to 0}\phi(\nu, u)=\infty$, and $\psi(\nu, u)=0$ for all $u\geq \beta\sup_{\theta\in\Theta}\lambda(\theta)\theta$.

\end{enumerate}
Consequently, $\phi(\nu, u)\leq \sup_{\nu'\in \Delta(\Theta)}\phi(\nu', u)\leq g(1/u)<\infty$ for all $\nu\in\supp(\sigma_\mu)$, which in turn implies $\phi(\cdot, u)\in \mathcal{L}^1(\Delta(\Theta), \sigma_\mu)$ for all $u>0$. 

Next, observe that a pair $(\tau^*_\sigma, u^*_\sigma)$ jointly satisfies  \eqref{eq:inequality} and \eqref{eq:equil} if and only if $u^*_\sigma$ is a fixed point of the function $\Phi:(0, \infty)\to \mathbb{R}_+$ given by 
\[
\Phi(u)\coloneqq k\int_{\Delta(\Theta)}m(\phi(\nu, u))\mathbb{E}_\nu[\lambda(\theta)\theta]\sigma_\mu(d\nu).
\]
The function $\Phi(u)$ is a continuous and weakly decreasing function because $\phi(\nu, u)$ is continuous and weakly decreasing in $u$ for each $\nu\in\supp(\sigma_\mu)$. Moreover, $\Phi(u)=0$ for all $u\geq \beta\sup_{\theta\in\Theta}\lambda(\theta)\theta$, and 
\begin{align*}
    \lim_{u\to 0}\Phi(u)&=k\int_{\Delta(\Theta)} \lim_{u\to 0}m(\phi(\nu, u))\mathbb{E}_\nu[\lambda(\theta)\theta]\sigma_\mu(d\nu)\\[6pt]
    &=\alpha k\mathbb{E}_{\mu}\big[\lambda(\theta)\theta\big]\\[6pt]
    &>0,
\end{align*}
where the first equality follows from the dominated convergence theorem, the second equality follows because $\lim_{u\to 0}\phi(\nu, u)=\infty$ for all $\nu$ such that $\mathbb{E}_{\nu}[\lambda(\theta)\theta)]>0$, and because $\lim_{t\to \infty}m(t)=\alpha$, and the inequality follows from \autoref{ass:2}.

Thus far, we have shown that $\Phi(u)$ is a continuous and (weakly) decreasing function with $\lim_{u\to \beta\sup_{\theta\in\Theta}\lambda(\theta)\theta}\Phi(u)=0<\lim_{u\to 0}\Phi(u)$. 
Therefore, it has a unique fixed point in the interval $(0, \beta\sup_{\theta\in\Theta}\lambda(\theta)\theta)$.  We have thus established the existence of an essentially-unique pair $(\tau^*_\sigma, u^*_\sigma)$ with $u^*_\sigma=\Phi(u^*_\sigma)$ and $\tau^*_\sigma(\cdot)=\phi(\cdot, u^*_\sigma)$, such that $(\tau^*_\sigma, u^*_\sigma)$ jointly satisfies \eqref{eq:inequality} and \eqref{eq:equil}. In other words, there exists an essentially-unique search equilibrium $(\tau^*_\sigma, u^*_\sigma)$ for each market segmentation $\sigma$.
\end{proof}\medskip

\begin{proof}[Proof of \autoref{thm:hosios}]\
(\emph{Only-if} direction):  Suppose the pair $(\sigma^{PS}, \tau^*_{\sigma^{PS}})$ solves \eqref{eq:planner}. By \autoref{thm:planner}, $\tau^*_{\sigma^{PS}}=\tau^{FB}_{\sigma^{PS}}$ for $\sigma^{PS}_\mu$-almost everywhere, or equivalently, $\tau^*_{\sigma^{PS}}(\delta_\theta)=\tau^{FB}_{\sigma^{PS}}(\delta_\theta)$ for $\mu$-almost all $\theta$. 

From \autoref{thm:planner}, $\tau^{FB}_{\sigma^{PS}}(\delta_\theta)=0$ for all $\theta\leq \eta_{\sigma^{PS}}/\beta$, and $\tau^{FB}_{\sigma^{PS}}(\delta_\theta)>0$ for all $\theta> \eta_{\sigma^{PS}}/\beta$. Similarly, from \autoref{thm:equil}, $\tau^*_{\sigma^{PS}}(\delta_\theta)=0$ for all $\theta$ such that $\lambda(\theta)\theta\leq u^*_{\sigma^{PS}}/\beta$, and $\tau^*_{\sigma^{PS}}(\delta_\theta)>0$ for all $\theta$ such that $\lambda(\theta)\theta> u^*_{\sigma^{PS}}/\beta$.  Since the set of active and inactive submarkets coincide almost everywhere under the two submarket tightness functions, we have that $\theta\leq \eta_{\sigma^{PS}}/\beta$ if and only if $\lambda(\theta)\theta\leq u^*_{\sigma^{PS}}/\beta$. By the continuity of the surplus-splitting function, we can conclude that 
\[
u^*_{\sigma^{PS}}=\lambda\left(\frac{\eta_{\sigma^{PS}}}{\beta}\right)\eta_{\sigma^{PS}}.
\]
Thus, for all $\theta\leq \eta_{\sigma^{PS}}/\beta$, 
\[
\lambda(\theta)\theta\leq \lambda\left(\frac{\eta_{\sigma^{PS}}}{\beta}\right)\frac{\eta_{\sigma^{PS}}}{\beta}.\]
Moreover, for all $\theta>\eta_{\sigma^{PS}}/\beta$, \autoref{thm:planner} implies that 
\[
m'(\tau^{FB}_{\sigma^{PS}}(\delta_\theta))\theta=\eta_{\sigma^{PS}},
\]
while \autoref{thm:equil} implies that
\[
\frac{m(\tau^*_{\sigma^{PS}}(\delta_\theta))}{\tau^*_{\sigma^{PS}}(\delta_\theta)}\lambda(\theta)\theta=\lambda\left(\frac{\eta_{\sigma^{PS}}}{\beta}\right)\eta_{\sigma^{PS}}.
\] 
 Combining the above two expressions and noting that $\tau^{FB}_{\sigma^{PS}}=\tau^*_{\sigma^{PS}}$ almost everywhere yields 
\[
\lambda(\theta)=\lambda\left(\frac{\eta_{\sigma^{PS}}}{\beta}\right)\cdot\varepsilon\left(f\left(\frac{\theta}{\eta_{\sigma^{PS}}}\right)\right)
\]
for all $\theta>\eta_{\sigma^{PS}}/\beta$. 
\medskip

(\emph{If} direction):  Suppose the surplus-splitting function, $\lambda$, satisfies the properties given in the proposition. We first show that  $(\tau^{FB}_{\sigma^{PS}}, \lambda({\eta_{\sigma^{PS}}}/{\beta})\cdot \eta_{\sigma^{PS}})$ is a search equilibrium of $\sigma^{PS}$. 

To that end, for $\theta\leq \eta_{\sigma^{PS}}/\beta$, we have 
\[
U(\delta_\theta;\tau^{FB}_{\sigma^{PS}}, \sigma^{PS})=\underbrace{\frac{m(0)}{0}}_{=\beta}\lambda(\theta)\theta\leq\lambda\left(\frac{\eta_{\sigma^{PS}}}{\beta}\right)\eta_{\sigma^{PS}},
\]
where the equality follows from the fact that $\tau^{FB}_{\sigma^{PS}}(\delta_\theta)=0$ for $\theta\leq \eta_{\sigma^{PS}}/\beta$ (\autoref{thm:planner}), and the inequality follows by the assumption on $\lambda$ given in the proposition. 

For $\theta>\eta_{\sigma^{PS}}/\beta$, we have 
\begin{align*}
U(\delta_\theta;\tau^{FB}_{\sigma^{PS}}, \sigma^{PS})&=\frac{m\left(f\left(\frac{\theta}{\eta_{\sigma^{PS}}}\right)\right)}{f\left(\frac{\theta}{\eta_{\sigma^{PS}}}\right)}\lambda(\theta)\theta\\[8pt]
&=m'\left(f\left(\frac{\theta}{\eta_{\sigma^{PS}}}\right)\right)\lambda\left(\frac{\eta_{\sigma^{PS}}}{\beta}\right)\theta\\[8pt]
&=\lambda\left(\frac{\eta_{\sigma^{PS}}}{\beta}\right)\eta_{\sigma^{PS}},
\end{align*}
where the first equality follows from the characterization of $\tau^{FB}_{\sigma^{PS}}(\delta_\theta)$ for $\theta> \eta_{\sigma^{PS}}/\beta$ (\autoref{thm:planner}), the second equality follows by the assumption on $\lambda$ given in the proposition, and the last equality follows from the fact that $f$ is the inverse of the mapping $t\mapsto 1/m'(t)$. Therefore, the pair $(\tau^{FB}_{\sigma^{PS}}, \lambda(\eta_{\sigma^{PS}})\eta_{\sigma^{PS}})$ satisfies \eqref{eq:inequality} of \autoref{def:equilibrium}. 

Additionally,
\begin{align*}
    k\int_{\Delta(\Theta)}m\big(\tau^{FB}_{\sigma^{PS}}(\delta_\theta)\big)\mathbb{E}_\nu[\lambda(\theta)\theta]\sigma_\mu^{PS}(d\nu)=& k\int_{\eta_{\sigma^{PS}}/\beta}^1 m\left(f\left(\frac{\theta}{\eta_{\sigma^{PS}}}\right)\right)\lambda(\theta)\theta\mu(d\theta)\\[8pt]
     =&\lambda\left(\frac{\eta_{\sigma^{PS}}}{\beta}\right)\eta_{\sigma^{PS}} \cdot k\int_{\eta_{\sigma^{PS}}/\beta}^1  f\left(\frac{\eta_{\sigma^{PS}}}{\theta}\right)\mu(d\theta)\\[8pt]
     =&\lambda\left(\frac{\eta_{\sigma^{PS}}}{\beta}\right)\eta_{\sigma^{PS}},
\end{align*}
where the first equality follows from the definition of $\sigma^{PS}$ and the characterization of $\tau^{FB}_{\sigma^{PS}}$ in \autoref{thm:planner}, the second equality follows by the assumption on $\lambda$ given in \autoref{thm:hosios} and the fact that $f$ is the inverse of the mapping $t\mapsto 1/m'(t)$, and the last equality follows by the construction of $\eta_{\sigma^{PS}}$ as given in \autoref{thm:planner}. Thus, the pair $(\tau^{FB}_{\sigma^{PS}}, \lambda({\eta_{\sigma^{PS}}}/{\beta})\cdot\eta_{\sigma^{PS}})$ also satisfies \eqref{eq:equil} of \autoref{def:equilibrium}.

We can therefore conclude that $(\tau^{FB}_{\sigma^{PS}}, \lambda({\eta_{\sigma^{PS}}}/{\beta})\cdot\eta_{\sigma^{PS}})$ is a search equilibrium of $\sigma^{PS}$. However, $(\tau^*_{\sigma^{PS}}, u^*_{\sigma^{PS}})$ is the essentially-unique search equilibrium of $\sigma^{PS}$. Thus, $\tau^*_{\sigma^{PS}}=\tau^{FB}_{\sigma^{PS}}$ for $\sigma_\mu^{PS}$-almost everywhere, and $u^*_{\sigma^{PS}}=\lambda({\eta_{\sigma^{PS}}}/{\beta})\cdot\eta_{\sigma^{PS}}$. Since $(\sigma^{PS}, \tau^{FB}_{\sigma^{PS}})$ solves \eqref{eq:planner}, the almost-everywhere equivalence between $\tau^*_{\sigma^{PS}}$ and $\tau^{FB}_{\sigma^{PS}}$ implies that $(\sigma^{PS}, \tau^*_{\sigma^{PS}})$ also solves \eqref{eq:planner}, giving us the desired result.
\end{proof}\medskip

\begin{proof}[Proof of \autoref{thm:secondbest}]\ 
Let us first establish some facts that will be useful in the proof: Under \autoref{ass:3}, $\beta\sup_{\theta\in\Theta}\lambda(\theta)\theta=\beta\ell$. With some abuse of notation, for $x\in\Theta$ and $u>0$, let 
\[
\phi(x, u)\coloneqq \left\{\begin{array}{ccc}0 & \mbox{if} & \ell x\leq u/\beta\\[4pt]
\displaystyle g\left(\frac{\ell x}{u}\right) & \mbox{if} & \ell x >u/\beta
  \end{array}\right.
  \]
which is simply a reformulation of the function $\phi(\nu, u)$ as defined in the proof of \autoref{thm:equil} into the current environment where payoffs depend only on the posterior mean. 

Given $u>0$, consider the following linear program: 
\begin{align*}
\label{eq:lp}
\tag{LP-$u$}
\max_{\widehat H\in \mpc(F)}\hspace*{.5em}k\int_{\Theta}\phi(x, u)\hspace{.1em}  d \widehat H(x).
\end{align*}
Let ${V}(u)$ be the value function corresponding to \eqref{eq:lp}, and let $\mathcal{H}(u)$ be the set of maximizers. Since the objective is a continuous function and $\mpc(F)$ is compact, $V(u)$ is continuous in $u$, and $\mathcal{H}(u)$ is non-empty and compact-valued for each $u>0$. Importantly, by construction, $H$ satisfies \eqref{eq:simplify} if and only if $H\in\mathcal{H}(u_H^*)$. 

Observe that $V$ is a (weakly) decreasing function as $\phi(x, u)$ is a (weakly) decreasing function of $u$ for all $x\in\Theta$. Furthermore, for almost all $x\in\Theta$, $\phi(x, u)$ grows arbitrarily large as $u\to 0$ and $\phi(x, u)=0$ for all $u\geq \beta\ell$, which implies that $\lim_{u\to 0}V(u)=\infty$ and $V(u)=0$ for all $u\geq \beta\ell$. Thus, there is a unique point $\bar{u}\in (0,\beta\ell)$ such that $V(\bar u)=1$. 
\medskip

\begin{lemma}\label{lemma:3}
Consider a market segmentation $H\in \mpc(F)$, and let $(\tau_H^*,  u_H^*)$ be its essentially-unique search equilibrium. Then $u_H^*\leq \bar u$. Furthermore, $H\in\mathcal{H}(u^*_H)$ if and only if $u_H^*=\bar u$. 
\end{lemma}
\begin{proof}[Proof of \autoref{lemma:3}]\
For each $H\in \mpc(F)$,
\[
V(\bar u)=1=k\int_\Theta \phi(x, u^*_H)d H(x) \leq  V(u^*_{H}),
\]
where the first equality follows by construction of $\bar u$, the second equality follows from \eqref{eq:feasible}, and the inequality follows because $V$ is the value function of \eqref{eq:lp}. Since $V$ is a weakly decreasing function, we have $\bar u\geq u_{H}^*$, establishing the first desired result.

Furthermore, by construction, the above inequality holds with equality if and only if $H\in\mathcal{H}(u^*_H)$. However, since $\bar u$ is the unique value for which $V(u)=1$, the above inequality holds with equality if and only $u_H^*=\bar u$. This establishes the second desired result. 
\end{proof}
\medskip

Let us now prove the general characterization for constrained-efficient market segmentations given in \eqref{eq:simplify}, and then prove the characterization for the special cases when $t/m(t)$ is convex or concave. 
\medskip

\noindent\textbf{General case:}\

(\emph{Only-if} direction): Suppose $H$ is a constrained-efficient market segmentation, i.e., $H$ solves \eqref{eq:sb2}. We want to show that $H\in\mathcal{H}(u^*_H)$. By \autoref{lemma:3}, this is equivalent to showing that $u_H^*=\bar u$. Furthermore, \autoref{lemma:3} already establishes that $\bar u\geq u_H^*$. Thus, to show the desired ``only if" direction, it suffices to show that $u_H^*\geq \bar u$. 

Notice that under \autoref{ass:3}, the buyers' payoff is equivalent to expected total surplus scaled down by a constant $\ell$. Therefore, because $H$ solves \eqref{eq:sb2}, we have $u^*_H\geq u^*_{\widehat H}$ for all $\widehat H\in\mpc(F)$. In particular---recalling that $\mathcal{H}(\bar u)$ is non-empty---for each $\widehat H\in\mathcal{H}(\bar u)$, we have $u_H^*\geq u^*_{\widehat H}=\bar u$, where the last equality follows from \autoref{lemma:3}. \medskip

(\emph{If} direction): Suppose $H$ satisfies \eqref{eq:simplify}, i.e., $H\in\mathcal{H}(u^*_H)$. By \autoref{lemma:3}, $\bar u=u_H^*\geq u_{\widehat H}^*$ for all $\widehat H\in\mpc(F)$. In other words, for all $\widehat H\in\mpc(F)$,
\[
\underbrace{\ell\cdot k\int_\Theta m(\tau_H^*(x))xdH(x)}_{=u_H^*}\geq  \underbrace{\ell\cdot k\int_\Theta m(\tau_{\widehat H}^*(x))xd\widehat H(x)}_{=u_{\widehat H}^*},
\]
where the equalities follow by \eqref{eq:equil}. The above inequality then establishes the desired result: $H$ solves \eqref{eq:sb2}. 

We have thus far shown that $H$ is a constrained-efficient market segmentation if and only if $H\in\mathcal{H}(\bar u)\equiv\mathcal{H}(u_H^*)$. This concludes the characterization for the general statement.
\medskip

Let us proceed by showing that when the mapping $t\mapsto t/m(t)$ is concave, then $F\in\mathcal{H}(u)$ for all $u\in(0,\beta\ell)$. Similarly, when the mapping $t\mapsto t/m(t)$ is convex, we show that for each $u\in(0,\beta\ell)$, there exists a binary market segmentation in $\mathcal{H}(u)$. 

First, suppose the mapping $t\mapsto t/m(t)$ is \textbf{concave}, and therefore, $g$ is convex. In this case, $\phi(x,  u)=g({\ell\cdot x}/{ u})\mathbbm{1}_{[\beta\ell x>  u]}$ is convex in $x$ for all $u\in(0,\beta\ell)$.\footnote{Recall that $\lim_{y\to 1/\beta}g(y)=0$.} As a result, for all  $H', H''\in\mpc(F)$ where $H'$ is a mean-preserving contraction of $H''$, we have
\[
k\int_\Theta \phi(x, u)dH'(x)\leq k\int_\Theta \phi(x, u)dH''(x).
\]
Therefore, $F\in\mathcal{H}(u)$ for all $u\in(0,\beta\ell)$. Consequently, $F\in\mathcal{H}(\bar u)$, i.e., there exists a perfect market segmentation is constrained-efficient, establishing Point $(i)$ of \autoref{thm:secondbest}.  \\

Second, suppose the mapping $t\mapsto t/m(t)$ is \textbf{convex}, and therefore, $g$ is concave. In this case,  $\phi(x,  u)=g({\ell\cdot x}/{ u})\mathbbm{1}_{[\ell\beta x>u]}$ is piece-wise concave in $x$ for all $u\in(0,\beta\ell)$. Specifically, it is constant over $[0,   u/(\beta\ell)]$ and concave over $[u/(\beta\ell), 1]$.\footnote{\cite{kol22} describe such a $\phi$ function as ``S-shaped." They show that if $\phi$ is continuously differentiable and S-shaped, there exists a solution to \eqref{eq:lp} that is an upper-censorship: perfectly segment types below a cutoff and pool all the types above the cutoff. However, their results cannot be directly applied here as the $\phi(x, u)$ is not differentiable at $x= \frac{u}{\beta\ell}$. Additionally, \autoref{thm:secondbest} goes further than an upper censorship; it claims that a binary segmentation is sufficient to obtain constrained efficiency.}

Define the functions $\underline X:\Theta\to\Theta$ and $\overline{X}:\Theta\to\Theta$ given by $\underline{X}(\theta)\coloneqq \mathbb{E}_F[\tilde \theta|\tilde\theta\leq \theta]$ and $\overline{X}(\theta)\coloneqq \mathbb{E}_F[\tilde \theta|\tilde\theta\geq \theta]$, respectively. As $F$ is assumed to be absolutely continuous, both $\underline X$ and $\overline{X}$ are continuous and strictly increasing functions. Furthermore, it can be readily checked that the mapping $\theta\mapsto \theta-\underline{X}(\theta)$ is positive and strictly increasing over $(0,1)$, while the mapping $\theta\mapsto \theta-\overline{X}(\theta)$ is negative and strictly increasing over $(0,1)$. 

For each $u\in(0,\beta\ell)$, let $\underline \theta_u\coloneqq \min\{\theta\in\Theta:\overline{X}(\theta)\geq u/(\beta\ell)\}$, which is well-defined because $\overline{X}$ is continuous and strictly increasing, and  $\{\theta\in\Theta:\overline{X}(\theta)\geq u/(\beta\ell)\}$ is non-empty.\footnote{For example, $\overline{X}(1)=1>u/(\beta\ell)$.} Define $G_u:[\underline \theta_u, 1]
\to \mathbb{R}$ as the function given by 
\[
G_u(\theta)\coloneqq g\left(\frac{\ell \cdot\overline{X}(\theta)}{ u}\right) +g'\left(\frac{\ell\cdot \overline{X}(\theta)}{ u}\right)\frac{\ell}{ u}\Big(\theta-\overline{X}(\theta)\Big),
\]
which is a continuous and strictly increasing function because $(a)$ both $g$ and $\overline{X}$ are continuous and strictly increasing functions, $(b)$  $g'$ is a (weakly) decreasing function by concavity of $g$, and $(c)$ the mapping $\theta\mapsto\theta-\overline{X}(\theta)$ is negative and strictly increasing over $(0,1)$. 

Next, define the cutoff $\theta^c_u\coloneqq \min\{\theta\in[\underline\theta_u, 1]:G_u(\theta)\geq 0\}$, which is well-defined because $G_u$ is continuous and strictly increasing, and $\{\theta\in[\underline\theta_u, 1]:G_u(\theta)\geq 0\}$ is non-empty.\footnote{For example, $1-\overline X(1)=0$, so $G_u(1)=g(\ell/u)>0$ for each $u\in(0,\beta\ell)$. The inequality follows because $g$ is a strictly increasing function with $g(1/\beta)=0$. }

\begin{lemma}\label{lemma:cutoffs}
For each $u\in(0,\beta\ell)$, 
\[
\underline X(\theta^c_u)\leq  \frac{u}{\beta\ell}<\overline{X}(\theta^c_u).
\]
\end{lemma}
\begin{proof}[Proof of \autoref{lemma:cutoffs}]
Fix an arbitrary $u\in(0,\beta\ell)$. First, notice that  $\overline{X}(u/(\beta\ell))>u/(\beta\ell)$ by construction. Since $\underline\theta_u$ is the smallest type in $\Theta$ that satisfies $\overline{X}(\theta)\geq u/(\beta\ell)$, we have $u/(\beta\ell)\geq \underline \theta_u$. 

Next, notice that 
\begin{align*}
G\left(\frac{u}{\beta\ell}\right)&=g\left(\frac{\ell\cdot \overline{X}\left(\frac{u}{\beta\ell}\right)}{u}\right) +g'\left(\frac{\ell\cdot \overline{X}\left(\frac{ u}{\beta\ell}\right)}{ u}\right)\frac{\ell}{ u}\Bigg(\frac{ u}{\beta \ell}-\overline{X}\left(\frac{ u}{\beta\ell}\right)\Bigg)\\[6pt]
    &\geq  g\left(\frac{\ell\cdot\frac{ u}{\beta\ell}}{ u}\right) \\
    &= 0,
\end{align*}
where the inequality follows from the concavity of $g$, and the last equality follows because $g(1/\beta)=0$. Since $\theta^c_u$ is the smallest type in $[\underline\theta_u, 1]$ that satisfies  $G_u(\theta)\geq 0$, we have $u/(\beta\ell)\geq \theta_u^c$. 

By construction, $\underline{X}(\theta)\leq \theta$ for all $\theta\in\Theta$. Thus, $\underline{X}(\theta^c_u)\leq \theta^c_u\leq u/(\beta\ell)$, which establishes the first inequality in the lemma.

Also by construction, $\theta^c_u\geq \underline \theta_u$, which implies that $\overline{X}(\theta^c_u)\geq \overline{X}(\underline\theta_u)\geq u/(\beta\ell)$. Suppose, for the sake of contradiction, $\overline{X}(\theta^c_u)= u/(\beta\ell)$. Then 
\begin{align*}
    G_u\left(\theta_u^c\right)=&\underbrace{g\left(\frac{\ell \cdot\frac{u}{\beta\ell}}{u}\right)}_{=0} +\underbrace{g'\left(\frac{\ell \cdot\frac{u}{\beta\ell}}{u}\right)\frac{\ell}{u}}_{>0}\underbrace{\Big(\theta_u^c-\overline{X}(\theta_u^c)\Big)}_{<0}
    <0,
\end{align*}
which contradicts the fact that, by definition, $G_u(\theta_u^c)\geq 0$. To avoid the contradiction, we must have $\overline{X}(\theta^c_u)>u/(\beta\ell)$, establishing the second inequality in the lemma.
\end{proof}

Finally, consider the binary posterior-mean distribution $H_u\in\mpc(F)$ given by 
 \[
   H_u(x)=\left\{\begin{array}{ccc}
0&\mbox{if}        &  x<\underline X(\theta^c_u)\\
F(\theta^c)&\mbox{if}        &  x\in [\underline X(\theta_u^c), \overline X(\theta_u^c))\\
1&\mbox{if}        &  x\geq \overline X(\theta^c_u)\\
   \end{array}\right. .
   \]  
   
From Theorem 1 of \cite{dwo19} (DM henceforth), $H_u$ solves \eqref{eq:lp} if there exists a convex function $p_u:\Theta\to \mathbb{R}$ such that 
\begin{enumerate}[$(\alph*)$]
    \item $\phi(\cdot, u)\leq p_u$ pointwise,
    \item $\supp(H_u)\subseteq \{x\in\Theta: \phi(x,  u)=p_u(x)\}$, and 
    \item $\mathbb{E}_{H_u}[p_u(x)]=\mathbb{E}_{F}[p_u(x)]$. 
\end{enumerate}

Consider the function $p:\Theta\to \mathbb{R}$ given by
\[
p_u(x)=\left\{\begin{array}{ccc}
  0 &  \mbox{if} & x< \theta^c_u  \\
 \displaystyle g\left(\frac{\ell\cdot \overline X(\theta^c_u)}{u}\right) +g'\left(\frac{\ell\cdot \overline X(\theta^c_u)}{u}\right)\frac{\ell}{u}\Big(x-\overline X(\theta^c_u)\Big)& \mbox{if}  &x\geq \theta^c_u
\end{array}\right..
\]
Notice that $p_u(x)$ is (weakly) increasing and convex as it is the upper envelope of two affine functions. Clearly, $p_u(x)=\phi(x, u)$ for all $x< \theta^c_u$. For $x\in[\theta^c_u, u/(\beta\ell)]$, we have 
\[
\phi(x, u)=0\leq G_u(\theta^c_u)=p_u(\theta^c_u)\leq p_u(x),
\]
where the first inequality follows by definition of $\theta^c_u$, the second equality follows by construction of $p_u$, and the second inequality follows because $p_u$ is weakly increasing.  For $x> u/(\beta\ell)$, we have 
\[
\phi(x, u)=g\left(\frac{\ell \cdot x}{u}\right)\leq g\left(\frac{\ell\cdot \overline X(\theta^c_u)}{u}\right) +g'\left(\frac{\ell\cdot \overline X(\theta^c_u)}{ u}\right)\frac{\ell}{ u}\Big(x-\overline X(\theta^c_u)\Big)=p_u(x),
\]
where the inequality follows from the concavity of $g$. Thus, $\phi(\cdot, u)\leq p_u$ pointwise, establishing point $(a)$ of DM.

By construction, 
\[
\supp(H_u)=\{\underline X(\theta_u^c), \overline X(\theta^c_u)\}\subseteq [0, \theta^c_u]\cup\{\overline X(\theta^c_u)\}\subseteq \{x\in\Theta:p_u(x)=\phi(x, u)\},
\]
establishing point $(b)$ of DM. 

Finally, $\mathbb{E}_{H_u}[p(x)]=p_u\big(\overline X(\theta^c_u)\big)(1-F(\theta^c_u))$ while
\begin{align*}
  \mathbb{E}_F[p_u(x)]=&\int^1_{\theta^c_u}\left[g\left(\frac{\ell \cdot \overline X(\theta^c_u)}{u}\right) +g'\left(\frac{\ell\cdot \overline X(\theta^c_u)}{u}\right)\frac{\ell}{u}\Big(x-\overline X(\theta^c_u)\Big)\right]dF(x)\\[8pt]
  =&p_u\big(\overline X(\theta^c_u)\big)(1-F(\theta^c_u)) +g'\left(\frac{\ell\cdot \overline X(\theta^c_u)}{u}\right)\frac{\ell}{ u}\underbrace{\int^1_{\theta^c_u} \Big(x-\overline X(\theta^c_u)\Big)dF(x)}_{=0}\\[6pt]
  =&p_u\big(\overline X(\theta^c_u)\big)(1-F(\theta^c_u)),
\end{align*}
where the last equality follows because, by definition, $\overline X(\theta^c_u)=\mathbb{E}_F[\tilde\theta|\tilde\theta\geq \theta^c_u]$, establishing point $(c)$ of DM.

We can therefore conclude that $H_u\in\mathcal{H}(u)$. Additionally, by \autoref{lemma:cutoffs}, $\phi(\underline X(\theta_u^c), u)=0<\phi(\overline{X}(\theta_u^c), u)$. Since these properties hold for each for each $u\in(0, \beta\ell)$, they also hold for $u=\bar u$. Therefore, there exists a constrained-efficient binary market segmentation in which all buyers join only the submarket with the highest expected type, while the submarket with the lowest expected type remains inactive in equilibrium. This establishes Point $(ii)$ of the proposition. 
\end{proof}

\subsection{General characterization for constrained-efficient  segmentations}\label{necessary}
This section provides a characterization for constrained-efficient market segmentations under \autoref{ass:meeting-function} and \autoref{ass:2} only.

From \autoref{thm:equil}, each market segmentation $\sigma$ has an essentially-unique search equilibrium $(\tau_\sigma^*, u^*_\sigma)$ with $\tau^*_\sigma(\cdot)=\phi(\cdot, u^*_\sigma)$, where recall that 
\[
\phi(\nu, u)=g\left(\frac{\mathbb{E}_\nu[\lambda(\theta)\theta]}{u}\right) \mathbbm{1}_{\left[\mathbb{E}_\nu[\lambda(\theta)\theta]>\frac{u}{\beta}\right]}.
\]

Given a market segmentation $\sigma$ and buyers' anticipated payoff $u$, define the ex-ante total surplus as
\[
\mathcal S(\sigma, u)\coloneqq k\int_{\Delta(\Theta)}m(\phi(\nu, u))\mathbb{E}_\nu[\theta]\sigma(d\nu),
\]
and the ex-ante buyer's payoff as 
\[
\mathcal{U}(\sigma, u)=k\int_{\Delta(\Theta)}m(\phi(\nu, u))\mathbb{E}_\nu[\lambda(\theta)\theta]\sigma(d\nu).
\]
We can then express the planner's second-best problem \eqref{eq:sb} as 
\[
\max_{\sigma} S(\sigma, u^*_\sigma).
\]
The objective in \eqref{eq:sb} is non-linear because the planner's choice of a market segmentation affects both what the buyers' believe about sellers' types within each submarket, and also how they trade off joining one submarket over another. Instead of directly solving \eqref{eq:sb} directly, consider instead the following problem:
\begin{align*}\label{eq:bp}
\tag{BP}
    \max_{\substack{\sigma,\\  u\in [0, \beta\sup_{\theta\in\Theta}\lambda(\theta)\theta]}} \mathcal{S}(\sigma, u) \hspace{1em} \text{s.t.} \hspace{1em} \mathcal{U}(\sigma, u)\leq u. 
\end{align*}
The problem in \eqref{eq:bp} decouples the two interdependent effects in \eqref{eq:sb}: holding fixed the anticipated payoff, the market segmentation now only affects beliefs, and  holding fixed the market segmentation, the anticipated payoff affects the trade-offs across submarkets. The interdependence between these two effects is then captured by the constraint, $\mathcal{U}(\sigma, u)\leq u$, which ensures that buyers' anticipated payoffs are consistent with their expected gains from trade. As such, holding $u$ fixed, the optimal choice of market segmentation in \eqref{eq:bp} is akin to a (constrained) Bayesian Persuasion problem with a continuum of states \cite{kam11}.

\begin{proposition}\label{thm:neccessary}
A market segmentation $\sigma$ is constrained-efficient if and only if $(\sigma, u^*_\sigma)$ is a solution to \eqref{eq:bp}.
\end{proposition}
\begin{proof}[Proof of \autoref{thm:neccessary}]
Let $\sigma'$ be a constrained-efficient market segmentation, i.e., $\sigma'$ is a solution to \eqref{eq:sb}. Let $(\sigma'', u'')$ be a solution to \eqref{eq:bp}.

Notice that for any market segmentation $\sigma$, the pair $(\sigma, u^*_\sigma)$ satisfy the constraint in \eqref{eq:bp}, and thus, \eqref{eq:bp} is a relaxation of \eqref{eq:sb}. Therefore, $\mathcal{S}(\sigma', u^*_{\sigma'})\leq \mathcal{S}(\sigma'', u'')$. 

Conversely, notice that $u''>0$; otherwise the constraint in \eqref{eq:bp} would be violated. Additionally, we must have $\mathcal{U}(\sigma'', u'')= u''$; otherwise, the planner could lower the anticipated payoff by a small amount, which would improve the objective without violating the constraint. However, the constraint holding with equality implies $u''=u^*_{\sigma''}$. Hence, $\mathcal{S}(\sigma', u^*_{\sigma'})\geq \mathcal{S}(\sigma'', u^*_{\sigma''})= \mathcal{S}(\sigma'', u'')$. 
\end{proof}

\newpage
\bibliographystyle{plainnat}
\nocite{}\bibliography{bibref}
\end{document}